\documentclass{article}
\usepackage[utf8]{inputenc}
\usepackage{amsmath,amssymb,amsthm}
\usepackage{bbm}
\usepackage{xcolor}
\usepackage{verbatim}
\usepackage{caption}
\usepackage{subcaption}
\usepackage{graphicx}
\usepackage{authblk}
\usepackage[colorlinks=true,linkcolor=blue]{hyperref}

\usepackage[left=2.5cm,right=2.5cm,top=2.5cm,bottom=2.5cm]{geometry}

\newtheorem{theorem}{Theorem}

\newtheorem{proposition}[theorem]{Proposition}
\newtheorem{lemma}[theorem]{Lemma}

\newtheorem{remark}[theorem]{Remark}

\newcommand{\A}{\mathbf{A}}
\newcommand{\B}{\mathbf{B}}
\newcommand{\C}{\mathbf{C}}
\newcommand{\D}{\mathbf{D}}
\newcommand{\e}{\mathbf{e}}
\newcommand{\M}{\mathbf{M}}
\newcommand{\PP}{\mathbb{P}}
\newcommand{\QQ}{\mathbb{Q}}

\renewcommand{\u}{\mathbf{u}}

\renewcommand{\v}{\mathbf{v}}

\newcommand{\x}{\mathbf{x}}
\newcommand{\X}{\mathbf{X}}
\newcommand{\y}{\mathbf{y}}
\newcommand{\Y}{\mathbf{Y}}
\newcommand{\z}{\mathbf{z}}
\newcommand{\Z}{\mathbf{Z}}

\newcommand{\GAM}{\boldsymbol{\Gamma}}
\newcommand{\ALPH}{\boldsymbol{\Phi}}
\newcommand{\BET}{\boldsymbol{\Psi}}

\newcommand{\Nc}{\mathcal{N}}

\newcommand{\DElta}{{\boldsymbol{\delta}}}
\newcommand{\eps}{\varepsilon}
\newcommand{\id}{\mathbf{Id}}

\newcommand{\SIGMA}{\boldsymbol{\Sigma}}
\newcommand{\SIGMADITH}{\tilde{\boldsymbol{\Sigma}}^{\operatorname{dith}}}
\newcommand{\Tau}{\boldsymbol{\tau}}
\newcommand{\ttheta}{\boldsymbol{\theta}}
\newcommand{\XI}{\boldsymbol{\Xi}}
\newcommand{\0}{\boldsymbol{0}}

\newcommand{\R}{\mathbb{R}}
\newcommand{\W}{\boldsymbol{W}}

\newcommand{\E}{\mathbb{E}}

\renewcommand{\P}[2][]{\;\mathbb{P}_{#1}\!\left[#2\right]}
\newcommand{\pnorm}[2]{\left\| #1 \right\|_{#2}}
\newcommand\tnorm[1]{\vert\vert\vert\mskip2mu
#1\mskip2mu \vert\vert\vert}
\newcommand{\abs}[2]{\left| #1 \right|_{#2}}

\newcommand{\round}[1]{\left( #1 \right)}

\newcommand{\curly}[1]{\left\{ #1 \right\} }

\newcommand{\diag}{\mathrm{diag}}
\newcommand{\dx}[1]{\; \mathrm{d} #1}
\newcommand{\sign}{\mathrm{sign}}

\newcommand{\tr}{\mathrm{Tr}}
\newcommand{\bE}{\mathbb{E}}

\definecolor{tumor}{HTML}{E37222}
\definecolor{tumred}{RGB}{205,32,44}
\definecolor{tumblue}{RGB}{00,115,207}
\definecolor{tumgreen}{RGB}{162,173,000}

\newcommand{\sjoerd}[1]{{{#1}}}
\newcommand{\sjoerdred}[1]{{{#1}}}
\newcommand{\sjoerdgreen}[1]{{{#1}}}
\newcommand{\johannes}[1]{{#1}}

\newcommand{\rev}[1]{#1}
\renewcommand{\ae}[1]{#1}

\allowdisplaybreaks

\title{Covariance Estimation under One-bit Quantization}
\author[$\star$]{Sjoerd Dirksen}
\author[$\dagger$]{Johannes Maly}
\author[$\ddag$]{Holger Rauhut}
\affil[$\star$]{\normalsize Mathematical Institute, Utrecht University, The Netherlands}
\affil[$\dagger$]{Department of Scientific Computing, KU Eichstaett/Ingolstadt, Germany}
\affil[$\ddag$]{Chair for Mathematics of Information Processing, RWTH Aachen University, Germany}
\date{}

\begin{document}

\maketitle

\begin{abstract}
\noindent We consider the classical problem of estimating the covariance matrix of a subgaussian distribution from i.i.d.\ samples in the novel context of coarse quantization, i.e., instead of having full knowledge of the samples, they are quantized to one or two bits per entry. This problem occurs naturally in signal processing applications. We introduce new estimators in two different quantization scenarios and derive non-asymptotic estimation error bounds in terms of the operator norm. In the first scenario we consider a simple, scale-invariant one-bit quantizer and derive an estimation result for the correlation matrix of a centered Gaussian distribution. In the second scenario, we add random dithering to the quantizer. In this case we can accurately estimate the full covariance matrix of a general subgaussian distribution by collecting two bits per entry of each sample. In both scenarios, our bounds apply to masked covariance estimation. We demonstrate the near-optimality of our error bounds by deriving corresponding (minimax) lower bounds and using numerical simulations.  
\end{abstract}

\section{Introduction}
\label{sec:Introduction}

Estimating covariance matrices of high-dimensional probability distributions from a finite number of samples is a core problem in multivariate statistics with numerous applications in, e.g., financial mathematics, pattern recognition, signal processing, and signal transmission \sjoerd{\cite{dahmen2000structured,haghighatshoar2018low,krim1996two,ledoit2003improved}}. In the classical problem setup, the task is to estimate the covariance matrix $\johannes{\SIGMA= \E(\X\X^T)} \in \R^{p\times p}$ of a \sjoerd{mean-zero} random vector $\X \in \R^p$ using $n$ i.i.d.\ samples $\X^1,...,\X^n \overset{\mathrm{d}}{\sim} \X$. The sample covariance matrix
\begin{align} \label{eq:SampleMean}
\hat{\SIGMA}_n = \frac{1}{n} \sum_{k=1}^n \X^k (\X^k)^T
\end{align}
is a natural estimator \sjoerd{as it converges a.s.\ to $\SIGMA$ for $n \rightarrow \infty$.} From a practical perspective, this asymptotic result is of limited use: it provides no information on the number of samples needed to guarantee \ae{$\tnorm{\hat{\SIGMA}_n - \SIGMA} < \eps$}, for a desired $\eps > 0$, where the error is measured in an appropriate norm \ae{$\tnorm{\cdot}$} (the most common choices being \sjoerd{the} operator and Frobenius norm\sjoerd{s}). In the last two decades, numerous works on \sjoerd{the} non-asymptotic analysis of covariance estimation showed that reliable approximation of $\SIGMA$ by $\hat{\SIGMA}_n$ becomes feasible for subgaussian distributions if $n \gtrsim p$. \sjoerd{For instance, if} $\X$ \sjoerd{has} a Gaussian distribution, \sjoerd{then} it is well known (see e.g.\ \cite{vershynin2018high}) that with probability at least $1-2e^{-t}$
\begin{align} \label{eq:Estimation_Gauss}
\pnorm{ \hat{\SIGMA}_n - \SIGMA }{} \lesssim \pnorm{\SIGMA}{} \round{ \sqrt{\frac{p + t}{n}} + \frac{p + t}{n} },
\end{align}
where \ae{$\|\cdot\|$ is the operator norm and} $a \lesssim b$ denotes $a \le C b$, for a certain absolute constant $C > 0$. Similar error bounds can be obtained for heavy-tailed distributions using more sophisticated estimat\sjoerd{ors} \cite{adamczak2010quantitative,adamczak2011sharp,cai2010optimal,koltchinskii2017,mendelson2018robust,srivastava2013covariance,vershynin2012close}. The sufficient number of samples for reliable estimation can be further reduced under suitable priors on $\SIGMA$, such as sparsity, low-rankness, or Toeplitz-structure, see e.g.\ \cite{bickel2008covariance,cai2013optimal,chen2012masked,kabanava2017masked,ke2019user,levina2012partial}.\par
In several concrete applications, the assumption that one has direct access to the samples $\X^k$ may be unrealistic. Especially in applications related to signal processing, samples are collected via sensors and hence need to be quantized to finitely many bits before they can be digitally transmitted and used for estimation. It is therefore no surprise that engineers have been examining the influence of rough quantization on correlation and covariance estimation for decades, see e.g., \cite{bar2002doa,choi2016near,jacovitti1994estimation,li2017channel,roth2015covariance}. However, in contrast to classical covariance estimation from \sjoerd{``unquantized''} samples, to the best of our knowledge only asymptotic estimation guarantees have been derived in the quantized setting so far.\par 
In this paper we start to close \sjoerd{this} gap by introducing estimation procedures for covariance estimation from quantized samples and deriving non-asymptotic bounds on the estimation error. For the sake of conciseness, we choose to focus on the more challenging problem of estimating the operator norm, although the methods in this paper can also directly be used to derive Frobenius norm bounds. We focus on two quantization models that are motivated by large-scale signal processing applications such as massive MIMO \cite{haghighatshoar2018low}, where high\sjoerd{-}precision quantization is infeasible in terms of expense and energy consumption. The memoryless one-bit quantizer $Q(x)=\sign(x)$ featured in both quantization models can be implemented in practice in a cheap and efficient manner. For this reason, it has been studied extensively in the recent years in the compressed sensing literature, see e.g.\ \cite{BFN17,DiM18a,DiM18b,GeS19,JLB13,KSW16,PlV13,PlV13lin} and the survey \cite{Dir19}.\par
\rev{
\sjoerd{In our analysis} we allow to incorporate priors on $\SIGMA$ in the form of a symmetric mask $\M \in [0,1]^{p\times p}$. This model was first suggested in \cite{levina2012partial} to unify existing sparse covariance estimation approaches, such as banding and tapering (see, e.g., \cite{bickel2008regularized,furrer2007estimation}). To be more precise, for a fixed mask $\M$ and an estimator $\SIGMA_n$, our results bound the masked error $\pnorm{\M \odot \SIGMA_n - \M \odot \SIGMA}{}$ in terms of $\M \odot \SIGMA$ (here $\odot$ denotes the entry-wise matrix product). By the triangle inequality
\begin{equation*}
\|\M \odot \SIGMA_n - \SIGMA\|\le \|\M \odot \SIGMA_n - \M \odot \SIGMA\| + \| \M \odot \SIGMA - \SIGMA \|
\end{equation*}
and we thus obtain estimation rates for $\pnorm{\M \odot \SIGMA_n - \SIGMA}{}$ that are up to log-factors independent of the ambient dimension $p$ whenever $\pnorm{\M \odot \SIGMA - \SIGMA}{}$ is small and $\M \odot \SIGMA$ is intrinsically low-dimensional, e.g., if $\SIGMA$ is sparse and $\M$ encodes its support. In practice, the exact support of $\SIGMA$ is unknown and it is a key challenge to select a suitable mask. In certain applications, it is known that the entries of $\SIGMA$ decay away from the diagonal, so that banding and tapering masks are suitable. Beyond the setting of this paper, it is in practice natural to use a data-dependent mask, e.g., by using adaptive thresholds (see, e.g., \cite{el2008operator}). Following \cite{chen2012masked,levina2012partial}, we will focus on the setting with a fixed mask.\par
Bounds for the masked error have been previously derived for the sample covariance matrix $\hat{\SIGMA}_n$. In particular, the authors of \cite{chen2012masked} showed for Gaussian distributions that
\begin{align} \label{eq:Estimation_Masked}
\johannes{\round{ \E \pnorm{ \M \odot \hat{\SIGMA}_n - \M \odot \SIGMA }{}^2 }^{\frac{1}{2}}}
\lesssim \pnorm{\M}{1\rightarrow 2} \sqrt{ \frac{ \pnorm{\SIGMA}{} \pnorm{\SIGMA}{\infty} \log(p)}{n} } + \pnorm{\M}{} \pnorm{\SIGMA}{\infty} \frac{\log(p)}{n},
\end{align}
where $\pnorm{\cdot}{\infty}$ denotes the \ae{entry-wise max-norm (i.e., the $\ell_\infty$-norm of the entries)} and $\pnorm{\cdot}{1\rightarrow 2}$ the maximum column $\ell_2$-norm. If $\M$ encodes the support of a covariance matrix that is sparse or of banded width, then the right hand side of \eqref{eq:Estimation_Masked} can be considerably smaller than the right hand side of \eqref{eq:Estimation_Gauss}. The estimate \eqref{eq:Estimation_Masked} will serve as a benchmark for our results. We refer to \cite{chen2012masked,levina2012partial} for a more extensive discussion of masked covariance estimation and its connections to established regularization techniques.}

\rev{
\subsection{Notation}

Before presenting our main results, it will convenient to discuss the notation that will be used throughout this work. We \sjoerd{write} $[n] = \{ 1,...,n \}$ for $n\in \mathbb{N}$. \sjoerd{We use the notation $a \lesssim_{\alpha} b$ (resp.\ $\gtrsim_{\alpha}$) to abbreviate $a \le C_{\alpha}b$ (resp.\ $\ge$), for a constant $C_{\alpha} > 0$ depending only on $\alpha$. Similarly, we write $a \lesssim b$ if $a \le Cb$ for an absolute constant $C>0$.} \sjoerdgreen{We write $a\simeq b$ if both $a\lesssim b$ and $b\lesssim a$ hold (with possibly different implicit constants)}. Whenever we use absolute constants $c,C > 0$, their values may vary from line to line. We denote the \sjoerd{all ones}-matrix by $\boldsymbol{1} \in \R^{p\times p}$. We let scalar-valued functions act entry-wise on vectors and matrices. In particular,
\begin{align*}{}
[\sign(\x)]_i = \begin{cases}
1 & \text{if } x_i\geq 0 \\
-1& \text{if } x_i<0,
\end{cases}
\end{align*}
for all $\x\in \R^p$ and $i\in [p]$. 
\sjoerd{For symmetric \johannes{$\W,\Z\in \R^{p\times p}$ we write $\W \preceq \Z$ if $\Z-\W$} is positive semidefinite. We will use that
	\begin{align} \label{eq:Kadison}
	\johannes{\mathbb{E}\Z^T\mathbb{E}\Z \preceq \mathbb{E}(\Z^T\Z) \quad \text{for any } \Z\in \R^{p_1\times p_2}}{},
	\end{align}
	which is immediate from 
	$$\0 \preceq \mathbb{E}[(\Z-\mathbb{E}\Z)^T(\Z-\mathbb{E}\Z)],$$ 
and refer to it as Kadison's inequality (as it is a special case of Kadison's inequality in noncommutative probability theory).}\par
For $\Z \in \R^{p\times p}$, we denote the operator norm by $\pnorm{\Z}{}~=~\sup_{\u \in \mathbb{S}^{p-1}} \pnorm{\Z\u}{2}$ (i.e., it is the maximum singular value), the entry-wise max-norm by $\pnorm{\Z}{\infty} = \max_{i,j} \abs{Z_{i,j}}{}$, and the maximum column norm \sjoerd{by} $\pnorm{\Z}{1\rightarrow 2} = \max_{j \in [p]} \pnorm{\z_j}{2}$, where $\z_j$ denotes the $j$-th column of $\Z$. We denote the Hadamard (i.e., entry-wise) product of two matrices by $\odot$ and define
\begin{align*}
\Z^{\odot \ell} = \underbrace{\Z \odot \cdots \odot \Z}_{\ell \text{-times}}.
\end{align*} 
The $\diag$-operator, when applied to \sjoerd{a} matrix, extracts the diagonal as \sjoerd{a} vector; when applied to a vector, it outputs the corresponding diagonal matrix. The subgaussian ($\psi_2$-) and subexponential ($\psi_1$-) norms of a random variable $X$ are defined by 
\begin{align*}
\pnorm{X}{\psi_2} = \inf \curly{ t>0 \colon \johannes{\E\round{ \exp\round{\frac{X^2}{t^2}} }} \le 2 }
\end{align*}
and
\begin{align*}
\pnorm{X}{\psi_1} = \inf \curly{ t>0 \colon \johannes{\E\round{ \exp\round{\frac{|X|}{t}} }} \le 2 }.
\end{align*}
Finally, a mean-zero random vector $\X$  \sjoerd{in $\R^p$} is called $K$-subgaussian if 
$$\|\langle \X,\x\rangle\|_{\psi_2} \leq K \johannes{(\mathbb{E}\langle \X,\x\rangle^2)^{1/2}} \quad \mbox{ for all }  \x \in \sjoerd{\R^p}.$$
}

\subsection{One-bit quantization}

\ae{We will now present our main results. Throughout, we will assume that $\X$ is mean-zero. Moreover, for our first model we will assume that $\X$ is Gaussian. In the first quantization model, we have access to i.i.d.\ one-bit quantized samples $\sign(\X^k) \in \curly{-1,1}^p$, for $k \in [n]$. Since the quantizer is scale-invariant, i.e., $\sign(\z) = \sign(\D\z)$ for any diagonal matrix $\D\in \R^{p\times p}$ with strictly positive entries and $\z \in \R^p$, we can only hope to recover the correlation matrix of the distribution in this setting. Hence, we assume that $\X \sim \mathcal{N}(\0,\SIGMA)$, where $\SIGMA$ has ones on its diagonal.}

As an estimator for $\SIGMA$ we consider
\begin{align} \label{eq:OneBitEstimator}
\tilde{\SIGMA}_n = \sin\left( \frac{\pi}{2n} \sum_{k=1}^n \sign(\X^k) \sign(\X^k)^T \right).
\end{align}
Its specific form is motivated by Grothendieck's identity (see, e.g., \cite[Lemma 3.6.6]{vershynin2018high}), also known as the ``arcsin-law" in the engineering literature \cite{jacovitti1994estimation,van1966spectrum}, which implies that    
\begin{align}{} \label{eq:Grothendieck}
\johannes{\E \round{\sign(\X^k) \sign(\X^k)^T}} = \frac{2}{\pi} \arcsin (\SIGMA).
\end{align}
Combined with the strong law of large numbers and the continuity of the sine function, this immediately shows that $\tilde{\SIGMA}_n$ is a consistent estimator of $\SIGMA$.\par
\sjoerdgreen{Our first main result is a non-asymptotic error bound for masked estimation with $\tilde{\SIGMA}_n$.}
\begin{theorem} \label{thm:Operator}
	\sjoerd{There exist absolute constants $c_1,c_2 > 0$ such that the following holds.} Let $\X \sim \Nc(\0,\SIGMA)$ with $\Sigma_{i,i} = 1$, for $i \in [p]$, and \sjoerd{$\X^1,...,\X^n \overset{\mathrm{i.i.d.}}{\sim} \X$}. Let $\M \in [0,1]^{p\times p}$ be a fixed symmetric mask. \sjoerd{If} $t \ge 0$ \sjoerd{and $n\geq c_1 \sjoerdgreen{\log^2(p)}(\log(p)+t)$,} 
	\sjoerd{then with probability} at least $1 - \sjoerd{2}e^{-c_2 t}$
	\begin{align}
	\label{eqn:OperatorEst}
	    &\pnorm{\M \odot \tilde{\SIGMA}_n - \M \odot  \SIGMA}{} \nonumber\\
	    &\quad \lesssim  \sjoerd{\|\sigma(\M \odot\A)\|}   \sqrt{\frac{\sjoerd{\log(p)} + t}{n}} + \max\curly{ \pnorm{ \M \odot \A}{}, \pnorm{ \M \odot \SIGMA}{} }  \frac{\sjoerd{\log(p)} + t}{n},
	\end{align}
	\ae{where $\A = \cos(\arcsin(\SIGMA))$ and $\sigma(\Z)$ is defined for symmetric $\Z$ by
	\begin{align*}
	\sigma(\Z)^2 := \frac{2}{\pi} \Z^2 \odot \arcsin(\SIGMA) - \frac{4}{\pi^2} \round{ \Z \odot \arcsin(\SIGMA) }^2.
	\end{align*}{}
	}
\end{theorem}

\sjoerd{We prove Theorem~\ref{thm:Operator} by writing $\M \odot \tilde{\SIGMA}_n - \M \odot \SIGMA$ in the form 
	\begin{equation}
	\label{eqn:realAnalyticDiff}
	f\left(\frac{1}{n}\sum_{k=1}^n \W_k\right)-f\left(\mathbb{E}\left(\frac{1}{n}\sum_{k=1}^n \W_k\right)\right),
	\end{equation}
	where $\W_k\in \R^{p\times p}$ are independent random matrices and $f$ is real analytic. By applying a Taylor series expansion, the latter can be written as
	$$\sum_{\ell=0}^{\infty} \C_{\ell}\odot \left(\frac{1}{n}\sum_{k=1}^n (\W_k-\johannes{\mathbb{E}\W_k} )\right)^{\odot \ell}.$$
	The main technical work is to estimate the terms of this series and we refer to the start of Section~\ref{sec:proofNoDith} for a more detailed proof sketch. \johannes{Along the way} we develop several tools that will be useful to estimate the operator norm of \eqref{eqn:realAnalyticDiff} for other real analytic $f$ and random matrices $\W_k$ than the ones considered here.}  
	
The \sjoerdgreen{first term on the right hand side of \eqref{eqn:OperatorEst}, featuring the factor $\sjoerd{\|\sigma(\M \odot\A)\|}$,} is rather uncommon in view of existing covariance estimation results, cf.\ \eqref{eq:Estimation_Gauss}, but not a proof artefact: \ae{the following lower bound for the expected squared error verifies that it is indeed the leading error term.
\begin{proposition} \label{prop:LB}
	\sjoerd{There exist constants $c_1,c_2,c_3>0$ such that the following holds. If $n\geq c_1\sjoerdgreen{\log^3(p)}$, $\X \sim \Nc(\0,\SIGMA)$ with $\Sigma_{i,i} = 1$, for $i \in [p]$, $\X^1,...,\X^n \overset{\mathrm{i.i.d.}}{\sim} \X$, and $\M \in [0,1]^{p\times p}$ is a fixed symmetric mask, then}
	\begin{align*}
	\johannes{ \round{ \E\pnorm{ \M \odot \tilde{\SIGMA}_n - \M \odot \SIGMA}{}^2}^{\frac{1}{2}} } 
	&\geq \sjoerd{\frac{c_2}{\sqrt{n}} \|\sigma(\M \odot\A)\|} + \frac{c_2}{n} \|\M \odot \SIGMA \odot \round{\boldsymbol{1} - \GAM^{\odot 2} }\| \\
	&  + \frac{c_2}{ n }\|\sigma(\M \odot \SIGMA)^2 \odot \GAM\|^\frac{1}{2} - \sjoerd{\max\{\|\A\|,\|\SIGMA\|\} \left(\frac{c_3\sjoerdgreen{\log(p)}}{n}\right)^{\frac{3}{2}}},
	\end{align*}
	where $\A$ and $\sigma$ are defined in Theorem \ref{thm:Operator} and 
	\begin{align*}
	    \GAM = \johannes{\E \round{\sign(\X) \sign(\X)^T}} = \frac{2}{\pi} \arcsin(\SIGMA).
	\end{align*}
\end{proposition}
The appearance of this unusual leading term} suggests that the \ae{one-bit estimator $\tilde{\SIGMA}_n$} may outperform the sample covariance matrix in cases where the coordinates of $\X$ are very strongly correlated. Indeed, \sjoerdgreen{since $A_{i,j}=(1 - \Sigma_{i,j}^2)^{1/2}$, the first term on the right hand side of \eqref{eqn:OperatorEst} completely vanishes} in \sjoerdgreen{the extreme} case of full correlation ($\SIGMA=\mathbf{1}$) and the second term\sjoerd{, in this case equal to} $p \frac{\sjoerd{\log(p)} + t}{n}$, becomes dominant. In the numerical experiments in Section~\ref{sec:Numerics} we verify this intuition.
Finally, let us note that in the \ae{unmasked case, i.e, if $\M = \mathbf 1$}, our estimator achieves the minimax rate up to the factor $\sjoerd{\log(p)}$, \sjoerd{see the discussion} in the Appendix.

\begin{remark}
	\label{rem:PSD}
	\sjoerd{The estimator $\M\odot\tilde{\SIGMA}_n$ may not be positive semidefinite. Indeed, this is already true in the unmasked case as applying the sine function element-wise to a matrix may fail to preserve positive semidefiniteness \cite{Scho42}. In practice, it may therefore be preferable to use $P_{\text{PSD}}(\M\odot \tilde{\SIGMA}_n)$ as an alternative estimator, where $P_{\text{PSD}}$ denotes the projection onto the positive semidefinite cone in terms of the operator norm. This positive semidefinite projection can easily be computed by performing an SVD, see \cite[Section 8.1.1]{boyd2004convex} for details. Since any convex projection is $1$-Lipschitz, \sjoerdgreen{Theorem~\ref{thm:Operator} yields the same estimate for the error $\|P_{\text{PSD}}(\M\odot \tilde{\SIGMA}_n)-\SIGMA\|$ of this alternative estimator.}}
\end{remark} 

\subsection{Quantization with dithering}

In the quantization model studied above, we were only able to estimate correlation matrices of centered Gaussian distributions. In our second quantization model, we aim to estimate the full covariance matrix of any centered subgaussian distribution. To achieve this, we introduce \emph{dithering} in the one-bit quantizer, i.e., we add random noise with a suitably chosen distribution to the samples before quantizing them. The insight that dithering can substantially improve the reconstruction from quantized observations goes back to the work of Roberts \cite{roberts1962picture} in the engineering literature (see also \cite{GrN98,GrS93}). In the context of one-bit compressed sensing, the effect of dithering was recently rigorously analyzed in \cite{BFN17,DiM18b,DiM18a,DMS22,JMP19,KSW16}.\par
To be able to reconstruct the full covariance matrix we collect two bits per entry of each sample vector and use independent uniformly distributed dithers when quantizing. Concretely, we have access to the quantized sample vectors 
$$\sign(\X^k + \Tau^k), \ \sjoerdgreen{\sign(\X^k + \bar{\Tau}^k)}, \qquad k=1,\ldots,n,$$
where the dithering vectors $\Tau^1,\bar{\Tau}^1,\ldots,\Tau^n,\bar{\Tau}^n$ are independent (and, moreover, independent of $\X^1,\ldots,\X^n$) and uniformly distributed in $[-\lambda,\lambda]^p$, with $\lambda>0$ to be specified later. Using these quantized observations, we construct the estimator
\begin{align} \label{eq:TwoBitEstimator}
\SIGMADITH_n = \frac{1}{2}\tilde{\SIGMA}'_n + \frac{1}{2}(\tilde{\SIGMA}'_n)^T
\end{align}
where
\begin{align} \label{eq:AsymmetricEstimator}
\tilde{\SIGMA}'_n = \frac{\lambda^2}{n}\sum_{k=1}^n \sign(\X^k + \Tau^k)\sign(\X^k + \bar{\Tau}^k)^T.
\end{align}
\ae{The estimator $\tilde{\SIGMA}'_n$ is quite different from $\tilde{\SIGMA}_n$ and seems to be more closely related to the sample covariance matrix. This stems from the nature of dithering. Indeed, as we show in Lemma~\ref{lem:Expectation}, if $\X$ takes values in $[-\lambda,\lambda]^p$, then 
$$\mathbb{E}_{\Tau^k,\bar{\Tau}^k}(\lambda^2\sign(\X^k + \Tau^k)\sign(\X^k + \bar{\Tau}^k)^T)=\X^k(\X^k)^T.$$
Hence, roughly speaking, dithering `cancels the quantization distortion' in expectation. As a consequence, $\tilde{\SIGMA}'_n$ is a consistent estimator in this case. If $\X$ is unbounded but concentrates around its mean, then $\tilde{\SIGMA}'_n$ is biased, but the bias can be controlled by setting $\lambda$ large enough. Let us note that the two-bit samples and the resulting asymmetric shape of $\tilde{\SIGMA}'_n$ (which is corrected by symmetrization in $\SIGMADITH_n$) are only needed to estimate the diagonal entries of $\SIGMA$: the off-diagonal entries can be estimated using one-bit samples $\sign(\X^k + \Tau^k)$, $k~=~1,\ldots,n$.}  

%

Our second main result quantifies the performance of $\SIGMADITH_n$ for masked covariance estimation of the full covariance matrix. 
\begin{theorem} \label{thm:OperatorDitheredMask}
	Let $\X$ be a mean-zero, $K$-subgaussian vector with covariance matrix \sjoerd{$\SIGMA$}. Let $\X^1,...,\X^n \overset{\mathrm{\sjoerd{i.i.d.}}}{\sim} \X$. Let $\M \in [0,1]^{p\times p}$ be a fixed symmetric mask. If $\lambda^2 \gtrsim_{\sjoerdgreen{K}} \log(n) \|\SIGMA\|_{\infty}$, then with probability at least $1-e^{-t}$,  
	$$\pnorm{\M \odot \SIGMADITH_n - \M \odot \SIGMA}{}\lesssim_{\sjoerdgreen{K}}
	\|\M\|_{1\to 2}(\lambda\|\SIGMA\|^{1/2}+\lambda^2)\sqrt{\frac{\log(p)+t}{n}} + \lambda^2\|\M\| \frac{\log(p)+t}{n}.$$
	In particular, if $\lambda^2 \simeq_{\sjoerdgreen{K}} \log(n) \|\SIGMA\|_{\infty}$, then
	\begin{align}
	\label{eqn:OperatorDitheredMask}
	    &\pnorm{\M \odot  \SIGMADITH_n - \M \odot \SIGMA}{} \nonumber\\
	    &\quad \lesssim_{\sjoerdgreen{K}}
	    \log(n) \|\M\|_{1\to 2}\sqrt{\frac{\|\SIGMA\| \ \|\SIGMA\|_{\infty}(\log(p)+t)}{n}} + \log(n)\|\M\|\|\SIGMA\|_{\infty}\frac{\log(p)+t}{n}. 
	\end{align}
\end{theorem}
\begin{remark}
\rev{In Theorem~\ref{thm:OperatorDitheredMask}, $\lambda$ depends on the number of samples $n$. In practice, it is preferable to design the quantization method without prior knowledge of the number of collected samples. This can be achieved using the following modifications. Let $\Tau^k,\bar{\Tau}^k$ instead be uniformly distributed in $[-\lambda_k,\lambda_k]^p$, where $\lambda_k\simeq_{\sjoerdgreen{K}} \log(k) \|\SIGMA\|_{\infty}$, and replace $\tilde{\SIGMA}'_n$ by 
$$\frac{1}{n}\sum_{k=1}^n \lambda_k^2 \ \sign(\X^k + \Tau^k)\sign(\X^k + \bar{\Tau}^k)^T.$$
One can show that the resulting estimator $\SIGMADITH_n$ satisfies \eqref{eqn:OperatorDitheredMask} by making straightforward modifications to the proof of Theorem~\ref{thm:OperatorDitheredMask}. For the sake of readability, we will stick to the setting of Theorem~\ref{thm:OperatorDitheredMask} in what follows.}     
\end{remark}
In addition to the previous remark, a remark analogous to Remark~\ref{rem:PSD} can be made in the context of Theorem~\ref{thm:OperatorDitheredMask}: $\M\odot \SIGMADITH_n$ is not positive semidefinite in general and hence it may be preferable in practice to use $P_{\operatorname{PSD}}(\M\odot \SIGMADITH_n)$ as an estimator. Theorem~\ref{thm:OperatorDitheredMask} immediately yields performance guarantees for this alternative estimator.\par
Surprisingly, the error bound \eqref{eqn:OperatorDitheredMask} matches the minimax rate up to logarithmic factors \ae{in the unmasked case (i.e., if $\M = \mathbf 1$)}, see the Appendix. Moreover, it is of the same shape (up to different logarithmic factors) as the best known estimate for the masked sample covariance matrix in \eqref{eq:Estimation_Masked}, even though the sample covariance matrix requires direct access to the \sjoerdgreen{``unquantized''} samples $\X^k$. \rev{In certain cases, however, the performance of $\SIGMADITH_n$ can be significantly worse than the performance of the sample covariance matrix. This can be seen more clearly in the unmasked case.} Indeed, in \cite{koltchinskii2017} it was shown that if the samples $\X^k$ are Gaussian, then 
$$\johannes{\E \|\hat{\SIGMA}_n - \SIGMA\|} \simeq \sqrt{\frac{\|\SIGMA\| \tr(\SIGMA)}{n}} + \frac{\tr(\SIGMA)}{n},$$
whereas \eqref{eqn:OperatorDitheredMask} yields 
$$\johannes{\E \|\SIGMADITH_n - \SIGMA\| } \lesssim
\log(n) \sqrt{\frac{p\|\SIGMA\| \ \|\SIGMA\|_{\infty}\log(p)}{n}} + \log(n)\frac{p\|\SIGMA\|_{\infty}\log(p)}{n}$$
via tail integration. \rev{The second estimate is clearly worse if $\tr(\SIGMA) \ll p\|\SIGMA\|_{\infty}$. Note that the latter condition implies that $\SIGMA$ has low effective rank $r(\SIGMA) = \tr(\SIGMA)/\|\SIGMA\| \ll p$, since $\|\SIGMA\|_{\infty} \le \|\SIGMA\|$.} In our numerical experiments we show that this difference is not a proof artefact, but rather a result of the distortion produced by the coarse quantization of the samples: whereas $\hat{\SIGMA}_n$ and $\SIGMADITH_n$ perform similarly for covariance matrices with a constant diagonal \rev{($\tr(\SIGMA)=p\|\SIGMA\|_{\infty}$)}, $\hat{\SIGMA}_n$ is observed to perform significantly better in situations where $\tr(\SIGMA)\ll p\|\SIGMA\|_{\infty}$. \ae{An intuitive explanation for this phenomenon is that in order to accurately estimate all diagonal entries of $\SIGMA$, the (maximal) dithering level $\lambda$ needs to be on the scale $\| \SIGMA \|_\infty$. If $\tr(\SIGMA)\ll p\|\SIGMA\|_{\infty}$, then most of the diagonal entries are much smaller than $\| \SIGMA \|_\infty$ and hence $\lambda$ is on a suboptimal scale for these entries.}

\sjoerd{The first part of the proof of Theorem~\ref{thm:OperatorDitheredMask} is to control the bias of the estimator $\M \odot \SIGMADITH_n$ using the tuning parameter $\lambda$. The proof can afterwards be finished using the matrix Bernstein inequality. This} proof strategy easily carries over to heavier-tailed random vectors. Nevertheless, the number of samples then increases since $\lambda$ has to be chosen larger. For instance, if $\X$ is a \sjoerdgreen{$K$-subexponential} random vector, one would already need $\lambda^2 \gtrsim_{\sjoerdgreen{K}} \log(n)^2\sjoerd{\|\SIGMA\|_{\infty}}$. The dependence of $\lambda$ on $n$, both in the latter statement and Theorem \ref{thm:OperatorDitheredMask}, is not an artifact of proof and observable in numerical experiments, \sjoerd{see} Section \ref{sec:Numerics}.

\subsection{Outline}

Section~\ref{sec:proofNoDith} is devoted to the proof of Theorem~\ref{thm:Operator}. In Section~\ref{sec:proofNoDithLower} we develop the corresponding lower bound for the estimation error \ae{stated in Proposition~\ref{prop:LB}}. In Section~\ref{sec:Dithered} we prove Theorem~\ref{thm:OperatorDitheredMask}. \ae{We conclude with numerical experiments in Section~\ref{sec:Numerics} and a discussion of our results in Section \ref{sec:Discussion}. A treatment of minimax bounds in the quantized setting is deferred to the Appendix.}

\section{Proof of Theorem~\ref{thm:Operator}}
\label{sec:proofNoDith}

Throughout our exposition, we \sjoerd{write} $Y_{i,j} = \sign(X_i)\sign(X_j)$, $Y_{i,j}^k = \sign(X_i^k)\sign(X_j^k)$, $\Y = \sign(\X)$, and $\Y^k = \sign(\X^k)$ \sjoerd{so} that the matrices $\Y\Y^T$ and $\Y^k (\Y^k)^T$ have entries $Y_{i,j}$ and $Y_{i,j}^k$. By using the Taylor series expansion
\begin{align*}
\sin(x) = \sin(a) + \cos(a) (x-a) - \frac{\sin(a)}{2} (x-a)^2 - \frac{\cos(a)}{6} (x-a)^3 + \cdots,
\end{align*}
for $x = \frac{\pi}{2n} \sum_{k=1}^n Y_{i,j}^k$ and $a = \frac{\pi}{2} \E Y_{i,j} = \arcsin(\Sigma_{i,j})$ (see \eqref{eq:Grothendieck} for the latter identity), we find
\begin{align}
\label{eqn:TaylorExpElementwise}
& (\tilde{\Sigma}_n)_{i,j} - \Sigma_{i,j} \nonumber\\
& \qquad = \cos\left( \frac{\pi}{2} \E{Y_{i,j}} \right) \left( \frac{\pi}{2n} \sum_{k=1}^n (Y_{i,j}^k - \E{Y_{i,j}}) \right)
- \frac{\sin \round{ \frac{\pi}{2} \E{Y_{i,j}} } }{2} \left( \frac{\pi}{2n} \sum_{k=1}^n (Y_{i,j}^k - \E{Y_{i,j}}) \right)^2 \nonumber\\
& \qquad \ \ \  - \frac{\cos \round{ \frac{\pi}{2} \E{Y_{i,j}} } }{6} \left( \frac{\pi}{2n} \sum_{k=1}^n (Y_{i,j}^k - \E{Y_{i,j}}) \right)^3
+ \cdots .
\end{align}
First note that $\cos \left( \frac{\pi}{2} \E{Y_{i,j}} \right) = \cos ( \arcsin ( \Sigma_{i,j})) =  \sqrt{1 - \Sigma_{i,j}^2} = \sjoerdgreen{A_{i,j}}$ and $\sin \left( \frac{\pi}{2} \E{Y_{i,j}} \right) = \Sigma_{i,j}$, for all $i,j \in [p]$. Let us define the \sjoerdgreen{random matrix} $\B \in \R^{p\times p}$ with entries $B_{i,j} = \frac{\pi}{2n} \sum_{k=1}^n(Y_{i,j}^k - \E{Y_{i,j}})$, for all $i,j \in [p]$, and note that
\begin{align}
\label{def:B}
\B = \frac{\pi}{2n} \sum_{k=1}^n \B_k,\quad \text{ where } \quad \B_k = \Y^k (\Y^k)^T - \johannes{\E \round{\Y^k (\Y^k)^T}}.
\end{align}
With this notation, the Taylor expansion \sjoerd{\eqref{eqn:TaylorExpElementwise}} yields
\begin{align} \label{eqn:TaylorExpMat}
\begin{split}
\tilde{\SIGMA}_n - \SIGMA
&= \A \odot \B - \SIGMA \odot \frac{1}{2} \B^{\odot 2} - \A \odot \frac{1}{6} \B^{\odot 3} + \cdots \\
&= \sum_{\ell = 0}^\infty (-1)^\ell \round{ \A \odot \B - \frac{1}{2(\ell + 1)} \SIGMA \odot \B^{\odot 2} } \odot \frac{1}{(2\ell + 1)!} \B^{\odot 2\ell}.
\end{split}
\end{align}
Before delving into the formal proof of Theorem~\ref{thm:Operator}, let us first sketch the steps that we will take. Each term in the series expansion \eqref{eqn:TaylorExpMat} has the form of a `Hadamard chaos'
\begin{equation}
\label{eqn:hadChaosTypical}
\C\odot \B^{\odot \theta} = \sum_{k_1,\ldots,k_{\theta}=1}^n \C\odot \B_{k_1}\odot \cdots\odot \B_{k_{\theta}}.
\end{equation}
We separately consider the cases $\theta \sjoerd{>} \log(p)$ and $\theta\leq \log(p)$. \rev{In the first case, it suffices to make the trivial estimate $\|\C\odot \B^{\odot \theta}\|\leq \|\C\| \ \|\B^{\odot \theta}\|_{1\to 2} \le \sqrt{p} \|\C\| \ \|\B^{\odot \theta}\|_{\infty}$ (see Lemma~\ref{lem:HadamardOperatorNorm}) and to estimate the $\ell_\infty$-norm on the right hand side via Bernstein's inequality. The additional $\sqrt{p}$-factor is outweighed by the Taylor-coefficients $\frac{1}{\theta !}$, for $\theta > \log(p)$. In the harder case $\theta~\leq~\log(p)$, this strategy does not work as it would lead to an additional (and suboptimal) dependence on $p$. Instead,} we will first show, using Lemma~\ref{lem:BkHadPowers}, that \eqref{eqn:hadChaosTypical} can be expressed as a sum of terms of the form
\begin{equation}
\label{eqn:HadChaosAllDifferentIndicesInt}
\sum_{k_1\neq \cdots\neq k_{\theta'}} \C'\odot \B_{k_1}\odot \cdots\odot \B_{k_{\theta'}},
\end{equation}
where $\theta'\leq \theta$ and $k_1\neq \cdots\neq k_{\theta'}$ means that we only sum terms for which all indices are different. Next, we show that terms of this form can be decoupled, i.e., we can replace $\B_{k_1},\ldots,\B_{k_{\theta'}}$ by copies \sjoerd{$\B_k^{(1)},...,\B_{k_{\theta'}}^{(\theta')}$} that are independent of each other (Lemma~\ref{lem:decoupling}). The resulting decoupled `Hadamard chaoses' can be estimated by iteratively applying the matrix Bernstein inequality (Lemma~\ref{lem:allDifferentIndices}). By using these steps and doing careful bookkeeping we \sjoerdgreen{arrive at a suitable estimate in Lemma~\ref{lem:HadPowersLessThanLogp}}.\par 
Let us now formally collect all ingredients for the proof. The first is a tool to control the operator norm of Hadamard products of matrices. It is a simple consequence of Schur's product theorem \cite[Eq.\ (3.7.11)]{johnson1990matrix}{}.
\begin{lemma} \label{lem:HadamardOperatorNorm}
	Let $\A,\B \in \R^{p\times p}$, let $\A$ be symmetric and define $\C = (\A^T\A)^\frac{1}{2}$. Then
	\begin{align*}
	\pnorm{\A \odot \B}{} \le \round{\max_{i \in [p]} C_{ii}} \pnorm{\B}{} = \pnorm{\A}{1\rightarrow 2} \pnorm{\B}{}.
	\end{align*}{}
	If $\A$ is in addition positive semidefinite, then
	\begin{align*}
	\pnorm{\A \odot \B}{} \le \round{\max_{i \in [p]} A_{ii} } \pnorm{\B}{}.
	\end{align*}
\end{lemma}{}
The next observation allows to reduce the problem of estimating a Hadamard chaos to estimating terms of the form \eqref{eqn:HadChaosAllDifferentIndicesInt}. 
\begin{lemma}
	\label{lem:BkHadPowers}
	\sjoerd{Let} $\GAM = \frac{2}{\pi} \arcsin(\SIGMA)$ \sjoerd{and} define 
	\begin{align*}
	\ALPH = \boldsymbol{1}-\GAM^{\odot 2} \quad \text{ and } \quad
	\BET = \sjoerd{-}2\GAM.
	\end{align*}
	Define the sequences $(\ALPH_n)_{n\geq 1}$ and $(\BET_n)_{n\geq 1}$ recursively by setting $\ALPH_1=\boldsymbol{0}$,  $\BET_1=\boldsymbol{1}$, and 
	\begin{align*}
	\ALPH_n=\ALPH\odot \BET_{n-1}, \qquad \BET_n = \ALPH_{n-1} + \BET\odot \BET_{n-1},
	\end{align*}
	for $n\geq 2$. Then, for $\B_k$ as \sjoerd{defined} in \eqref{def:B}
	and \sjoerd{any} $n\geq 1$,
	\begin{align*}
	\B_k^{\odot n} = \ALPH_n + \BET_n\odot \B_k.
	\end{align*}
	\sjoerdgreen{Moreover}, for any symmetric $\C\in \R^{p\times p}$,
	\begin{align*}
	\max\{\|\ALPH_n\odot \C\|,\|\BET_n\odot \C\|\}\leq [1+\|\GAM\|_{\infty}]^{2(n-1)}\|\C\| \le 4^{n-1} \pnorm{\C}{}.
	\end{align*}
\end{lemma}
\begin{proof}
	Note that $\ALPH_2 = \ALPH$ and $\BET_2 = \BET$ \johannes{and thus, by} \eqref{eq:Grothendieck},
	\begin{align*}
	\B_k^{\odot 2} &=
	(\Y^k (\Y^k)^T)^{\odot 2} - 2 \Y^k (\Y^k)^T \odot \johannes{\E \round{ \Y^k (\Y^k)^T } + \round{ \E \round{\Y^k (\Y^k)^T} }^{\odot 2}} \\
	&= \boldsymbol{1} - 2 \johannes{\round{\Y^k(\Y^k)^T - \E \round{\Y^k (\Y^k)^T}} \odot \E \round{\Y^k (\Y^k)^T} - \round{ \E \round{\Y^k (\Y^k)^T} }^{\odot 2}} \\
	&= \boldsymbol{1} - \frac{4}{\pi} \arcsin(\SIGMA) \odot \B_k - \frac{4}{\pi^2} \arcsin(\SIGMA)^{\odot 2} = \ALPH + \BET\odot \B_k,
	\end{align*}
	\johannes{i.e., the asserted formula holds for $n=2$. Furthermore, if it} holds for $n-1$, then
	\begin{align*}
	\B_k^{\odot n} 
	&= \B_k^{\odot (n-1)}\odot \B_k 
	= [\ALPH_{n-1} + \BET_{n-1}\odot \B_k]\odot \B_k\\
	& = \ALPH_{n-1} \odot \B_k + \BET_{n-1}\odot[\ALPH + \BET\odot \B_k] \\
	& = \BET_{n-1}\odot \ALPH + [\ALPH_{n-1} + \BET\odot \BET_{n-1}]\odot \B_k \\ 
	&= \ALPH_n + \BET_n\odot \B_k.
	\end{align*}
	The result \sjoerd{therefore} follows by induction.
	To prove the final statement, set $c_n~=~\|\BET_n~\odot~\C\|$,
	so in particular $c_1=\|\C\|$. Since $\GAM$ \sjoerd{is} the covariance matrix of $\sign(\X)$ \sjoerd{(see} \eqref{eq:Grothendieck}), it is positive semidefinite. Lemma \ref{lem:HadamardOperatorNorm} together with $2x \le (1+x)^2$, for $x \in \mathbb{R}$, \sjoerd{therefore} yields
	\begin{align*}
	\|\BET_2\odot \C\| 
	= 2\|\GAM\odot \C\|
	\leq 2\|\GAM\|_{\infty} \|\C\| 
	\le (1+\|\GAM\|_{\infty})^2 \|\C\|.
	\end{align*}
	Moreover, for any $n\geq 3$, by the triangle inequality and Lemma \ref{lem:HadamardOperatorNorm}
	\begin{align*}
	c_n &= \|\BET_n\odot \C\| 
	= \|\ALPH\odot\BET_{n-2}\odot \C + \BET\odot\BET_{n-1}\odot \C\| \\
	&\le \pnorm{\boldsymbol{1} \odot \BET_{n-2}\odot \C}{} + \pnorm{ \GAM^{\odot 2} \odot \BET_{n-2}\odot \C}{} + 2\pnorm{\GAM \odot \BET_{n-1}\odot \C}{} \\
	&\leq (1+\|\GAM\|_{\infty}^2)c_{n-2} + 2\|\GAM\|_{\infty}c_{n-1} \sjoerdgreen{\leq (1+\|\GAM\|_{\infty})^2 \max\{c_{n-2},c_{n-1}\}}.
	\end{align*}
	By induction, it immediately follows that, for any $n\geq 1$,
	$$c_n\leq (1+\|\GAM\|_{\infty})^{2(n-1)}\|\C\|.$$
	As a consequence, for any $n\geq 1$,
	\begin{align*}
	\|\ALPH_n\odot \C\| 
	&= \|\ALPH\odot \BET_{n-1}\sjoerd{\odot}\C\| 
	\leq (1+\|\GAM\|_{\infty}^2)c_{n-1} \leq (1+\|\GAM\|_{\infty}^2)(1+\|\GAM\|_{\infty})^{2(n-2)}\|\C\| \\
	&\leq (1+\|\GAM\|_{\infty})^{2(n-1)}\|\C\|,
	\end{align*}
	which concludes the proof.
\end{proof}
Thanks to Lemma~\ref{lem:BkHadPowers}, we can reduce the terms in \eqref{eqn:hadChaosTypical} to \sjoerdgreen{a sum of} terms of the form \sjoerd{\eqref{eqn:HadChaosAllDifferentIndicesInt}}, in which all indices are different. This allows us to make use of the following decoupling inequality. Let us note that a similar decoupling result, where the summation \sjoerd{on the right hand side of \eqref{eqn:decouplingLemmaEquation}} is restricted to terms with different indices, is a consequence of a general result of de la Pe\~{n}a (see \cite[Theorem 2]{pen92}). Our proof is a straightforward generalization of the proof of the well-known decoupling inequality for a second order chaos, due to Bourgain and Tzafriri~\cite{bourgain1987invertibility}. 
\begin{lemma} \label{lem:decoupling}
	Fix $\C\in \R^{p\times p}$ \sjoerd{and let $\B_1,\ldots,\B_n$ be independent, mean-zero random matrices in $\R^{p\times p}$}. Let $F\colon \R^{p\times p} \rightarrow \R$ be convex, \ae{let $\theta \geq 1$}, and let \sjoerdgreen{$(\B_k^{(1)})_{k=1}^n,...,(\B_k^{(\theta)})_{k=1}^n $ be independent and identically distributed with $(\B_k)_{k=1}^n$}. Then, for any $J \subset [n]$ with $|J| \ge \theta$,
	\begin{align}
	\label{eqn:decouplingLemmaEquation}
	\E{F \round{ \sum_{ k_1\neq\cdots\neq k_{\theta} \in J} \C\odot \B_{k_1}\odot \cdots \odot \B_{k_\theta} }}
	\leq \E{F \round{ C_\theta \sum_{ k_1,... , k_{\theta} \in J} \C\odot \B_{k_1}^{(1)}\odot \cdots \odot \B_{k_\theta}^{(\theta)} } },
	\end{align}
	where the sum on the left-hand side is over all tuples \sjoerd{consisting} of $\theta$ indices which all differ from each other, \ae{ $C_\theta = 1$ for $\theta = 1$, and $C_\theta = \theta! \; \Pi_{k=2}^\theta \round{\frac{k}{k-1}}^{k-1} \le \theta! \; e^\theta$ for $\theta\geq 2$.}
\end{lemma}
\begin{proof}
	We show the claim via induction. \ae{For $\theta = 1$ it \sjoerd{trivially} holds, for any $J \subset [n]$ and $C_{\theta} = 1$.} Let us now assume that the claim holds for $\theta-1$ and any $J \subset [n]$. We show that the claim holds as well for $\theta$. Fix any $J = \curly{j_1,...,j_{\abs{J}{}}} \subset [n]$. \sjoerd{We introduce} $\abs{J}{}$ i.i.d.\ Bernoulli random variables $\DElta = (\delta_{j_1},...,\delta_{j_{\abs{J}{}}})$ with $\P{\delta_{j_i} = 1} = 1 - \P{\delta_{j_i} = 0} = c$ and \sjoerd{choose $c$} such that, for any $k_1 \neq \cdots \neq k_\theta \in J$, $\johannes{\E\round{\delta_{k_1}\cdots \delta_{k_{\theta-1}} (1-\delta_{k_\theta}) }} = (1-c) c^{\theta-1}$ \sjoerd{is maximal, i.e., we set $c = \frac{\theta-1}{\theta}$}. By Jensen's inequality and Fubini's theorem,
	\begin{align*}
	&\E{F \round{ \sum_{ k_1\neq\cdots\neq k_{\theta} \in J} \C\odot \B_{k_1}\odot \cdots \odot \B_{k_\theta} }} \\
	&\quad = 
	\johannes{\E_\B} {F \round{ \theta \round{\frac{\theta}{\theta-1}}^{\theta-1} \sum_{ k_1\neq\cdots\neq k_{\theta} \in J} \johannes{\E_\DElta \round{\delta_{k_1}\cdots \delta_{k_{\theta-1}} (1-\delta_{k_\theta}) } \cdot } \ \C\odot \bigodot_{i=1}^\theta \B_{k_i} }} \\
	&\quad \le
	\johannes{\E_\DElta \E_\B} F \round{ \theta \round{\frac{\theta}{\theta-1}}^{\theta-1} \sum_{ k_1\neq\cdots\neq k_{\theta} \in J} \delta_{k_1}\cdots \delta_{k_{\theta-1}} (1-\delta_{k_\theta}) \round{ \C\odot \bigodot_{i=1}^\theta \B_{k_i} } } \\
	&\quad =
	\johannes{\E_\DElta \E_{\B,\B^{(\theta)}}} F \round{ \theta \round{\frac{\theta}{\theta-1}}^{\theta-1} \sum_{ k_1\neq\cdots\neq k_{\theta-1} \in J_{\DElta}} \sum_{ k_{\theta} \in \sjoerd{J\setminus J_{\DElta}}} \C\odot \bigodot_{i=1}^{\theta-1} \B_{k_i} \odot \B_{k_\theta}^{(\theta)} },
	\end{align*}
	where we denote $J_\DElta = \curly{j \in J \colon \delta_j = 1}$. \sjoerd{Hence, there must exist a realization $\DElta^*$ of $\DElta$ such that
		\begin{align*}
		&\E{F \round{ \sum_{ k_1\neq\cdots\neq k_{\theta} \in J} \C\odot \B_{k_1}\odot \cdots \odot \B_{k_\theta} }} \\
		&\quad \leq \johannes{\E_{\B,\B^{(\theta)}}} F \round{ \theta \round{\frac{\theta}{\theta-1}}^{\theta-1} \sum_{ k_1\neq\cdots\neq k_{\theta-1} \in I} \sum_{ k_{\theta} \in \sjoerd{J\setminus I }} \C\odot \bigodot_{i=1}^{\theta-1} \B_{k_i} \odot \B_{k_\theta}^{(\theta)} }
		\end{align*}
		with $I := J_{\DElta^*}$. Applying the induction hypothesis, we find}
	\begin{align*}
	&\E{F \round{ \sum_{ k_1\neq\cdots\neq k_{\theta} \in J} \C\odot \B_{k_1}\odot \cdots \odot \B_{k_\theta} }} \\
	&\quad \le
	\johannes{\E_{\B^{(\theta)}} \E_\B} F \round{ \theta \round{\frac{\theta}{\theta-1}}^{\theta-1} \sum_{ k_1\neq\cdots\neq k_{\theta-1} \in I} \sum_{ k_{\theta} \in \sjoerd{J\setminus I}} \C\odot \bigodot_{i=1}^{\theta-1} \B_{k_i} \odot \B_{k_\theta}^{(\theta)} } \\
	&\quad \le
	\johannes{\E_{\B^{(1)},...,\B^{(\theta)}}} F \round{ C_\theta \sum_{ k_1,\dots,k_{\theta-1} \in I} \sum_{ k_{\theta} \in \sjoerd{J\setminus I}} \C\odot \B_{k_1}^{(1)} \odot \cdots \odot \B_{k_{\theta-1}}^{(\theta-1)} \odot \B_{k_\theta}^{(\theta)} } \\
	&\quad = \johannes{\E_{\B^{(1)},...,\B^{(\theta)}}} F \left( C_\theta \sum_{ k_1,\dots, k_{\theta-1} \in I} \left( \sum_{ k_{\theta} \in \sjoerd{J\setminus I}} \C\odot \bigodot_{i=1}^{\theta} \B_{k_i}^{(i)}\right.\right. \\
	& \quad \qquad \qquad \qquad \qquad \qquad \qquad \qquad\qquad \qquad + \left.\left.\sum_{ k_{\theta} \in I} \C\odot \bigodot_{i=1}^{\theta-1} \B_{k_i}^{(i)} \odot \johannes{\E_{\B^{(\theta)}}\B_{k_\theta}^{(\theta)}}\right)\right)  \\
	&\quad \le \johannes{\E_{\B^{(1)},...,\B^{(\theta)}}} F \left( C_\theta \sum_{ k_1,\dots, k_{\theta-2} \in I} \sum_{ k_\theta \in J} \left( \sum_{ k_{\theta-1} \in I} \C\odot \bigodot_{i=1}^{\theta} \B_{k_i}^{(i)} \right. \right. \\
	& \quad \qquad \qquad \qquad \qquad \qquad\qquad \qquad\qquad \ \ \sjoerd{\left.\left.  + \sum_{ \ae{k_{\theta-1} \in J\setminus I} } \C\odot \bigodot_{i\neq \theta-1} \B_{k_i}^{(i)} \odot \johannes{\E_{\B^{(\theta-1)}} \B_{k_{\theta-1}}^{(\theta-1)}} \right)\right)} \\
	&\quad \quad \vdots \\
	&\quad \le
	\E{F \round{ C_\theta \sum_{ k_1,..., k_{\theta} \in J} \C \odot \B_{k_1}^{(1)} \odot \cdots \odot \B_{k_\theta}^{(\theta)} }}{},
	\end{align*}
	where we use the induction hypothesis in the second inequality to replace $\B_{k_1},...,\B_{k_{\theta-1}}$ by \sjoerdgreen{$\B_{k_1}^{(1)},...,\B_{k_{\theta-1}}^{(\theta-1)}$} and then apply $\theta$ consecutive steps of adding terms $\E \B_{k_j}^{(i)} = 0$ and applying Jensen's inequality. \sjoerd{The estimate of $C_\theta$ is a consequence of \johannes{$\round{(1+\tfrac{1}{k})^k}_{k\geq 1}$} being monotonically increasing with limit $e$.} 
\end{proof}
\johannes{The decoupled terms can be estimated by the following lemma. Recall that $\sigma(\Z)$ is defined for symmetric $\Z$ by}
\begin{align*}
\johannes{\sigma(\Z)^2 := \Z^2 \odot \GAM - \round{ \Z \odot \GAM}^2.}
\end{align*}{}
\begin{lemma} \label{lem:allDifferentIndices}
	Let $\X \sim \Nc(\0,\SIGMA)$ and $\X^1,...,\X^n \overset{\mathrm{i.i.d.}}{\sim} \X$. Define \sjoerd{$(\B_k)_{k=1}^n$} as in \eqref{def:B}
	and let \sjoerd{$(\B_k^{(1)})_{k=1}^n,...,(\B_k^{(\theta)})_{k=1}^n$ be independent and identically distributed with $(\B_k)_{k=1}^n$}. There exists an absolute constant $c>0$ such that for any $\theta\geq 1$, \sjoerd{any symmetric} $\C \in \R^{p\times p}$, and any $q\geq \log(p)$, 
	\begin{align*}
	\sjoerd{\johannes{\round{ \E \left\| \sum_{ k_1,...,k_{\theta} = 1}^n \C \odot \B_{k_1}^{(1)}\odot \cdots \odot \B_{k_\theta}^{(\theta)}\right\|^q}^{1/q}}} 
	\leq c^{\theta} (\sqrt{qn}+q)^{\theta-1} (\sqrt{qn}\|\sigma(\C)\|+q\| \C \|).
	\end{align*}
\end{lemma}
\johannes{To prove Lemma \ref{lem:allDifferentIndices} we will make use of the following $L^p$-version of the matrix Bernstein inequality \cite{tropp2012user}, see \cite[Theorem 6.2]{Dir14}, and the basic observation in Lemma \ref{lem:sigmavsOperatornorm} below.}
\begin{theorem}
	\label{thm:sumsRMs} 
	Let $2\leq q<\infty$. If $({\XI}_k)_{k=1}^n$ is a sequence of independent, mean-zero random matrices in $\R^{d_1\times d_2}$, then
	\begin{align*}
	\johannes{\Big(\E \Big\|\sum_{k=1}^n {\XI}_k\Big\|^q \Big)^{1/q}} & \leq C_{q,d} \max\Big\{\Big\|\Big(\sum_{k=1}^n \mathbb{E}(\XI_k^T\XI_k)\Big)^{1/2}\Big\|, \Big\|\Big(\sum_{k=1}^n \mathbb{E}(\XI_k\XI_k^T)\Big)^{1/2}\Big\|,\\
	& \qquad \qquad\qquad\qquad\qquad\qquad \qquad\ \ \ \ \ 2C_{\frac{q}{2},d} \johannes{\Big(\E \max_{1\leq k\leq n} \|{\XI}_k\|^q\Big)^{1/q}\Big\},}
	\end{align*}
	where $d=\max\{d_1,d_2\}$ and
	$C_{q,d}\leq  2^{\frac{3}{2}} e (1 + \sqrt{2}) \sqrt{\max\{q,\log d\}}$. 
\end{theorem}
\sjoerd{We will combine Theorem~\ref{thm:sumsRMs} with the following observation.}
\begin{lemma}
	\label{lem:sigmavsOperatornorm}
	\sjoerd{For any symmetric $\Z\in \R^{p\times p}$, $\|\sigma(\Z)\|\leq \sqrt{2}\|\Z\|$.}   
\end{lemma}
\begin{proof}
	\sjoerd{By the triangle inequality,
		$$\|\sigma(\Z)\|^2 \le \pnorm{ \Z^2 \odot \GAM }{} + \pnorm{ (\Z \odot \GAM )^2 }{} = \pnorm{ \Z^2 \odot \GAM }{} + \pnorm{\Z \odot \GAM }{}^2.$$
		Since $\GAM = \johannes{\E \round{\sign(\X)\sign(\X)^T}}$ is positive semidefinite, Lemma \ref{lem:HadamardOperatorNorm} immediately yields}
	$$\sjoerd{\|\sigma(\Z)\|^2\leq \sjoerdgreen{(\|\GAM\|_{\infty}+\|\GAM\|_{\infty}^2)}\pnorm{\Z}{}^2\leq 2\pnorm{\Z}{}^2.}$$
\end{proof}
\begin{proof}[Proof of Lemma \ref{lem:allDifferentIndices}]
	Recall $\Y = \sign(\X)$, $\Y^k = \sign(\X^k)$, and $\B^k=\Y^k(\Y^k)^T - \johannes{\E\round{\Y^k(\Y^k)^T}}{}$. Let us first note that Theorem~\ref{thm:sumsRMs} implies that, for any symmetric $\C_1,\ldots,\C_n\in \R^{p\times p}$ and $q \ge \log(p)$,
	\begin{align}
	\label{eqn:consNCBernstein}
	\sjoerd{ \johannes{\round{\E \pnorm{\sum_{k=1}^n \C_k \odot \B_k }{}^q}^{\frac{1}{q}}}{}} \lesssim \sqrt{qn}\round{\max_{k \in [n]} \|\sigma(\C_k)\|} + q \round{\max_{k \in [n]} \pnorm{\C_k}{}}.
	\end{align}
	Indeed, since 
	$$\sjoerd{\C_k \odot \Y^k(\Y^k)^T = \diag(\Y^k) \C_k \diag(\Y^k),}$$
	the triangle inequality, sub-multiplicativity of the operator norm, and $\pnorm{\diag(\Y^k)}{} =1$ yield
	\begin{align*}
	\pnorm{\C_k \odot \B_k}{} &= \pnorm{ \diag(\Y^k) \cdot \C_k \cdot \diag(\Y^k) - \johannes{\E \round{ \diag(\Y^k) \cdot \C_k \cdot \diag(\Y^k) }}{} }{} \\
	& \le 2 \ \sjoerd{\max_{k \in [n]} \pnorm{\C_k}{}}
	\end{align*}
	and
	\begin{align*}
	&\pnorm{ \sum_{k=1}^{n} \E{ (\C_k \odot \B_k)^2 } }{} \\
	&\quad \le n \max_{k \in [n]} \pnorm{\E{ (\C_k \odot \B_k)^2 }}{}
	= n \max_{k \in [n]} \pnorm{ \E{(\C_k \odot \Y^k(\Y^k)^T)^2} - \johannes{(\E{(\C_k \odot \Y^k(\Y^k)^T)})^2}{} }{} \\
	&\quad = n \max_{k \in [n]} \pnorm{ \C_k^2 \odot \johannes{\E \round{ \Y^k(\Y^k)^T }}{} - \round{ \C_k \odot \johannes{\E \round{ \Y^k(\Y^k)^T} } }^2 }{} = n \ \sjoerd{\max_{k \in [n]} \|\sigma(\C_k)^2\|}.
	\end{align*}
	In particular, by Theorem \ref{thm:sumsRMs} the result holds for $\theta=1$. \par
	Assume now \sjoerd{that} the asserted inequality holds for $\theta-1$. Applying \eqref{eqn:consNCBernstein} to the inner \sjoerd{expectation in} 
	\begin{align*}
	\sjoerd{ \johannes{\round{ \E \left\| \sum_{ k_1,...,k_{\theta} = 1}^n \C \odot \B_{k_1}^{(1)}\odot \cdots \odot \B_{k_\theta}^{(\theta)}\right\|^q}^{\frac{1}{q}}}}{} = \sjoerd{\johannes{\round{ \mathbb{E}_{\B^{(1)},...,\B^{(\theta-1)}}\mathbb{E}_{\B^{(\theta)}}\left\| \sum_{k=1}^n \bar{\C} \odot \B_{k}^{(\theta)}\right\|^q }^{\frac{1}{q}}}}{},
	\end{align*}
	\sjoerd{where
		\begin{align*}
		\bar{\C} = \sum_{ k_1,...,k_{\theta-1} = 1}^n \C \odot \B_{k_1}^{(1)}\odot \cdots \odot \B_{k_{\theta-1}}^{(\theta-1)},
		\end{align*}
	}and using that $\|\sigma(\bar{\C})^2\| \le 2 \pnorm{\bar{\C}}{}^2$ by Lemma~\ref{lem:sigmavsOperatornorm}, we immediately find, for $q \ge \log(p)$,
	\begin{align*}
	&\johannes{ \round{ \E\left\| \sum_{ k_1,...,k_{\theta} = 1}^n \C \odot \B_{k_1}^{(1)}\odot \cdots \odot \B_{k_\theta}^{(\theta)}\right\|^q}^{1/q}} \\
	&\quad \le
	c(\sqrt{qn}+q)
	\johannes{ \round{ \E \left\| \sum_{ k_1,...,k_{\theta-1} = 1}^n \C \odot \B_{k_1}^{(1)}\odot \cdots \odot \B_{k_{\theta-1}}^{(\theta-1)}\right\|^q}^{1/q}}.
	\end{align*}
	Applying the induction hypothesis to the right hand side yields the claim.
\end{proof}
\johannes{Combining the above results, we can control the moments of the `Hadamard chaos' in \eqref{eqn:hadChaosTypical}, for $\theta \le \log(p)$.}
\begin{lemma}
	\label{lem:HadPowersLessThanLogp}
	\sjoerd{Let $\B = \frac{\pi}{2n} \sum_{k=1}^n \B_k$, with $\B_k$ as defined in \eqref{def:B}. \sjoerdgreen{For any $2 \le \theta \le \log(p)\leq q \leq \frac{n}{\log^2(p)}$} and any symmetric $\C \in \mathbb{R}^{p\times p}$,    
		$$\frac{1}{\theta !}\johannes{\round{ \mathbb{E}\left\|\C \odot \B^{\odot \theta}\right\|^q }^{1/q}}{} \leq \sjoerdgreen{\round{\frac{c q}{n}}^{\theta/2}} \pnorm{\C}{}.$$
	}
\end{lemma}
\sjoerdred{
	\begin{proof}
		\sjoerdgreen{Throughout, $c$ will denote an absolute constant that may change from line to line.} \sjoerd{We decompose $\C \odot \B^{\odot \theta}$ into terms with all different indices, i.e.,}
		\begin{align}
		\label{eqn:decompCBtheta}
		\frac{1}{\theta !}\C \odot \B^{\odot \theta} \notag
		& = \frac{1}{\theta !}\round{\frac{\pi}{2n}}^{\theta} \sum_{m=0}^{\theta-2} \sum_{m'=1}^{\left\lfloor \frac{\theta-m}{2} \right\rfloor} \sum_{\substack{ \theta-m \ge j_1 \ge \cdots \ge j_{m'} \ge 2 \\ j_1 + \cdots + j_{m'} = \theta-m  }} \binom{\theta}{j_1,\dots,j_{m'},m} \notag \\
		& \qquad \qquad \qquad \qquad\qquad  \qquad\qquad \ \ \ \cdot \sum_{k_1 \neq \cdots \neq k_{m+m'}} \C \odot \bigodot_{i=1}^m \B_{k_i} \odot \bigodot_{i=1}^{m'} \B_{k_{m+i}}^{j_i} \notag\\
		& \qquad \qquad + \frac{1}{\theta !} \round{\frac{\pi}{2n}}^{\theta}\sum_{k_1 \neq \cdots \neq k_{\theta}} \C \odot \B_{k_1} \odot \cdots \odot \B_{k_\theta} =: \alpha + \beta.
		\end{align}
		\sjoerd{Note that the index $m$ tracks the number of $\B_k$ with power $1$ and $m'$ tracks the number of $\B_k$ with power at least $2$.} By Lemmas~\ref{lem:decoupling} and \ref{lem:allDifferentIndices}, 
		\begin{align*}
		(\mathbb{E}\|\beta\|^q)^{1/q}& =\frac{1}{\theta !} \round{\frac{\pi}{2n}}^{\theta} \johannes{\round{ \E \pnorm{ \sum_{k_1 \neq \cdots \neq k_{\theta}} \C \odot \B_{k_1} \odot \cdots \odot \B_{k_\theta}}{}^q}^{\frac{1}{q}}}{} \\
		& \leq e^{\theta} \round{ \round{\frac{cq}{n}}^{\frac{\theta}{2}} + \round{\frac{cq}{n}}^\theta } \pnorm{\C}{}.
		\end{align*}
		\sjoerd{It remains to estimate $(\mathbb{E}\|\alpha\|^q)^{1/q}$.} We use
		$$
		\round{ \sum_{m=0}^{\theta-2} \sum_{m'=1}^{\left\lfloor \frac{\theta-m}{2} \right\rfloor} \sum_{\substack{ \theta-m \ge j_1 \ge \cdots \ge j_{m'} \ge 2 \\ j_1 + \cdots + j_{m'} = \theta-m  }} 1 } \leq P(\theta) \leq e^{c \sqrt{\theta}},
		$$
		where the second inequality is a standard estimate (see \cite{apostol2013introduction}) for the partition function
		$$P(\theta) = | \{(j_1,j_2,\hdots,j_n) : n \in \mathbb{N}, j_1,\hdots,j_n \in \mathbb{N}, j_1\leq j_2 \leq \cdots \leq j_n; j_1+\hdots+j_n = \theta\} |,$$
		which counts the number of ways a natural number can be partitioned into a sum of natural numbers. This estimate implies
		\begin{align*}
		& (\mathbb{E}\|\alpha\|^q)^{1/q} \\
		& \ \leq c^\theta \frac{1}{\theta ! n^\theta} \cdot e^{c\sqrt{\theta}} \cdot \max_{\substack{m,m',\\(j_1,\dots ,j_{m'})}} \frac{\theta !}{j_1! \cdots j_{m'}! m!} \johannes{ \round{ \E \pnorm{ \sum_{k_1 \neq \cdots \neq k_{m+m'}} \C \odot \bigodot_{i=1}^m \B_{k_i} \odot \bigodot_{i=1}^{m'} \B_{k_{m+i}}^{j_i} }{}^q }^{\frac{1}{q}} }
		\end{align*}
		Applying Lemma \ref{lem:BkHadPowers} and recalling that $j_i \ge 2$, for all $i \in [m']$, we obtain
		\begin{align*}
		& (\mathbb{E}\|\alpha\|^q)^{1/q} \\
		& \ \leq \frac{c^\theta}{n^\theta} \max_{\substack{m,m',\\(j_1,\dots ,j_{m'})}} \frac{1}{2^{m'} m!} \johannes{ \round{ \E \pnorm{ \sum_{k_1 \neq \cdots \neq k_{m+m'}} \C \odot \bigodot_{i=1}^m \B_{k_i} \odot \bigodot_{i=1}^{m'} (\ALPH_{j_{i}} + \BET_{j_{i}} \odot \B_{k_{m+i}}) }{}^q }^{\frac{1}{q}} }.
		\end{align*}
		For $s\in [m']$, let us define
		\begin{align*}
		C_{\ALPH,\BET}^{s,m'} := \curly{ \bigodot_{i \in I^c} \ALPH_{j_i} \odot \bigodot_{i \in I} \BET_{j_i} \colon I \subset [m'], \abs{I}{} \le s }.
		\end{align*}
		Note that $\abs{C_{\ALPH,\BET}^{s,m'}}{} = \binom{m'}{s}$ and, by Lemma \ref{lem:BkHadPowers}, 
		\begin{equation}
		\label{eqn:BkHadPowerConsequenceProof}
		\pnorm{\C_{s} \odot \D}{} \le \round{1 + \pnorm{\GAM}{\infty}}^{2(\theta - m - \sjoerd{m'})} \pnorm{\D}{} \le 2^{2\theta} \pnorm{\D}{},
		\end{equation}
		for all $s \in [m']$, $\C_s \in C_{\ALPH,\BET}^{s,m'}$, and symmetric $\D \in \R^{p\times p}$. 
		We thus obtain with Lemmas \ref{lem:decoupling} and \ref{lem:allDifferentIndices} that
		\begin{align}
		\label{eqn:alphaPenultimateEst}
		& (\mathbb{E}\|\alpha\|^q)^{1/q} \nonumber\\
		&\quad \leq \frac{c^\theta}{n^\theta} \max_{\substack{m,m',\\(j_1,\dots ,j_{m'})}} \round{ \frac{1}{2^{m'} m!} \johannes{ \round{ \E \pnorm{ \sum_{k_1 \neq \cdots \neq k_{m+m'}} \sum_{s = 0}^{m'} \sum_{\C_{s} \in C_{\ALPH,\BET}^{s,m'}} (\C \odot \C_{s}) \odot \bigodot_{i=1}^{m+s} \B_{k_i} }{}^q }^{\frac{1}{q}} } }\nonumber\\
		& \quad = \frac{c^\theta}{n^\theta} \max_{\substack{m,m',\\(j_1,\dots ,j_{m'})}} \left( \frac{1}{2^{m'} m!}  \right. \nonumber \\
		& \qquad \qquad \qquad \left. \cdot \left( \E \pnorm{ \sum_{s = 0}^{m'} \sum_{\C_{s} \in C_{\ALPH,\BET}^{s,m'}} \frac{(n-(m+s))!}{(n-(m+m'))!} \sum_{k_1 \neq \cdots \neq k_{m+s}} (\C \odot \C_{s}) \odot \bigodot_{i=1}^{m+s} \B_{k_i} }{}^q \right)^{\frac{1}{q}} \right)  \nonumber\\
		&\quad \le \frac{c^\theta}{n^\theta} \max_{\substack{m,m',s,\C_s,\\(j_1,\dots ,j_{m'})}} \left( \frac{\sjoerd{m'} }{2^{m'} m!} \binom{m'}{s}  \frac{(n-(m+s))!}{(n-(m+m'))!} \right. \nonumber \\
		& \quad \qquad \qquad \qquad \qquad \qquad \qquad \left. \cdot \left( \E \pnorm{ \sum_{k_1 \neq \cdots \neq k_{m+s}} (\C \odot \C_{s}) \odot \bigodot_{i=1}^{m+s} \B_{k_i} }{}^q \right)^{\frac{1}{q}} \right) \nonumber\\
		&\quad \le \frac{c^\theta}{n^\theta} \theta \max_{\substack{m,m',s,\C_s,\\(j_1,\dots ,j_{m'})}} \left( \frac{1}{2^{m'} m!} \binom{m'}{s}  \frac{(n-(m+s))!}{(n-(m+m'))!} \sjoerd{(m+s)! e^{m+s}} \right. \nonumber\\
		& \quad  \qquad \qquad \qquad \qquad \qquad \qquad  \qquad \qquad \cdot \left. \johannes{ \round{ \E \pnorm{ \sum_{k_1,...,k_{m+s}=1}^n (\C \odot \C_{s}) \odot \bigodot_{i=1}^{m+s} \B_{k_i}^{(i)} }{}^q }^{\frac{1}{q}} } \right) \nonumber\\
		&\sjoerdgreen{\quad \le \frac{c^\theta}{n^\theta} \theta \max_{\substack{m,m',s,\C_s,\\(j_1,\dots ,j_{m'})}} \round{ \frac{1}{m!} \binom{m'}{s}  \frac{(n-(m+s))!}{(n-(m+m'))!} (m+s)! c^{m+s} (\sqrt{qn} + q)^{m+s} \pnorm{\C \odot \C_{s}}{}  }}\nonumber\\
		& \sjoerdgreen{\quad \leq \frac{c^\theta}{n^\theta} \max_{\substack{m,m',s,\\(j_1,\dots ,j_{m'})}} \round{ \frac{1}{m!} \binom{m'}{s}  \frac{(n-(m+s))!}{(n-(m+m'))!} (m+s)! (\sqrt{qn} + q)^{m+s} \pnorm{\C}{}  }},
		\end{align}        
		\sjoerd{where in the final step we used $\theta\leq e^{\theta}$ and \eqref{eqn:BkHadPowerConsequenceProof}. \sjoerdgreen{Since 
				$$\frac{(n-(m+s))!}{(n-(m+m'))!} (m+s)! = \frac{(m+m')!{ n\choose m+m'}}{{ n\choose m+s}},$$
				we find using the well-known estimates $k!\simeq k^{k+\frac{1}{2}}e^{-k}$ and $(\frac{k}{\ell})^{\ell}\leq {k\choose \ell}\leq (\frac{ek}{\ell})^{\ell}$ that
				\begin{align*}
				& \frac{1}{m!} \binom{m'}{s}  \frac{(n-(m+s))!}{(n-(m+m'))!} (m+s)! \\
				& \qquad \leq c^{\theta} \frac{1}{m^m} \left(\frac{m'}{s}\right)^s  \frac{(m+m')^{m+m'}\left(\frac{n}{m+m'}\right)^{m+m'}}{\left( \frac{n}{m+s}\right)^{m+s}}\\
				& \qquad = c^{\theta} \frac{1}{m^m}\left(\frac{m'}{s}\right)^s (m+s)^{m+s}n^{m'-s} \\
				& \qquad = c^{\theta}(m')^s \left(1+\frac{s}{m}\right)^m \left(1+\frac{m}{s}\right)^s n^{m'-s} \leq c^{\theta}\left(\frac{m'}{n}\right)^s \sqrt{n}^{2m'},
				\end{align*}
				where we repeatedly used $s,m',m\leq\theta$. Applying this in \eqref{eqn:alphaPenultimateEst} yields \begin{align*}
				(\mathbb{E}\|\alpha\|^q)^{1/q}
				& \leq \frac{c^\theta}{n^\theta} \max_{\substack{m,m',s,\\(j_1,\dots,j_{m'})}} 
				\left( \frac{(\sqrt{qn} + q)m'}{n}\right)^s \sqrt{n}^{2m'} (\sqrt{qn} + q)^{m} \pnorm{\C}{} \\
				& \leq c^{\theta}\round{ \sqrt{\frac{q}{n}} + \frac{q}{n}}^{\theta} \pnorm{\C}{},
				\end{align*}
				as $m+2m'\leq\theta\leq \log(p)$ and $n\geq q\log^2(p)\geq q(m'^2)$.
		}}
	\end{proof}
}

\sjoerd{To convert the $L^q$-bound \sjoerdgreen{in Lemma~\ref{lem:HadPowersLessThanLogp}} into a tail estimate, we use the following immediate consequence of Markov's inequality. 
	\begin{lemma}
		\label{lem:LqtoTailBound}
		Assume $0\leq q_0<q_1\leq \infty$. Let $\xi$ be a random variable satisfying
		$$(\mathbb{E}|\xi|^q)^{1/q}\leq a\sqrt{q}+bq$$
		for certain $a,b>0$ and all $q_0\leq q\leq q_1$. Then, 
		$$\johannes{\P{|\xi|\geq 2e\max\{a\sqrt{t},bt\}}}{} \leq e^{-t}$$
		for all $q_0\leq t\leq q_1$.
	\end{lemma}
	\begin{proof}
		Let $s>0$ be such that 
		$$2\max\{a \sqrt{q_0},b q_0\}\leq s\leq 2\max\{a\sqrt{q_1},b q_1\},$$
		then
		$$q:=\min\left\{\frac{s^2}{4a^2},\frac{s}{2b}\right\}$$
		satisfies $q_0\leq q\leq q_1$. By Markov's inequality,
		$$\johannes{\P{|\xi|\geq es}}{}
		\leq \frac{\mathbb{E}|\xi|^q}{(es)^q}\leq \left(\frac{a\sqrt{q}+b q}{es}\right)^q
		\johannes{\le \round{\frac{1}{e}}^{q} = }
		\exp\left(-\min\left\{\frac{s^2}{4a^2},\frac{s}{2b}\right\}\right).$$
		Setting $s=2\max\{a\sqrt{t},b t\}$ yields the result.
	\end{proof}
}Finally, we \sjoerd{will use the} following simple consequence of Bernstein's inequality. We include a proof for the sake of completeness.
\begin{lemma} \label{lem:MaxBound}
	There exist absolute constants $c_1,c_2>0$ such that the following holds. Suppose that $n\geq c_1\log(p)$. Let $\X \sim \Nc(\0,\SIGMA)$ with $\Sigma_{i,i} = 1$ for all $i \in [p]$ and let $\X^1,...,\X^n \overset{\mathrm{i.i.d.}}{\sim} \X$. Define $\B_k$ as in \eqref{def:B}.	Then
	\begin{align*}
	\P{ \pnorm{ \frac{1}{n} \sum_{k=1}^n \B_k }{\infty} \ge \sqrt{c_1 \frac{\log(p)}{n} + t} \ } \leq 2 e^{-c_2 n t}.
	\end{align*}{}
\end{lemma}
\begin{proof}
	\sjoerd{Recall the notation} $Y_{i,j}^k = \sign(X_i^k) \sign(X_j^k)$ and note that $(\B_k)_{i,j} = Y_{i,j}^k - \E{Y_{i,j}^{\sjoerd{k}}}$. Since $|Y_{i,j}^k - \E{Y_{i,j}^{\sjoerd{k}}}| \le 2$ for all $i,j,k$, the bound is trivial for $t\geq 4$. Using Bernstein's inequality for bounded random variables (see, e.g., \cite[Theorem 2.8.4]{vershynin2018high}) and $|Y_{i,j}^k - \E{Y_{i,j}^{\sjoerd{k}}}| \le 2$, we find, for any $u\leq 3$,
	\begin{align*}
	\P{ \frac{1}{n} \left| \sum_{k=1}^n Y_{i,j}^k - \E{Y_{i,j}^{\sjoerd{k}}} \right| \ge u } \le 2e^{-d_1 \min\left\{ \frac{n^2 u^2}{\sigma_{i,j}^2}, \frac{nu}{2} \right\}}
	\leq 2e^{-d_2 n\min\{u^2,u\}} \leq 2e^{-d_3 n u^2},
	\end{align*}
	as
	\begin{align*}
	\sigma_{i,j}^2 := \sum_{k=1}^n \E{(Y_{i,j}^k - \E{Y_{i,j}^{\sjoerd{k}}})^2} & = \sum_{k=1}^n \left( \E{(Y_{i,j}^k)^2} - \sjoerd{(
	}\E{Y_{i,j}^{\sjoerd{k}}}\sjoerd{)}^2 \right) \\
	& = n \cdot \left( 1 - \left( \frac{2}{\pi} \arcsin (\Sigma_{i,j}) \right)^2 \right) \le n.
	\end{align*}
	For any given $t<4$ and $n \ge d_4 \log(p)$ (with $d_4$ chosen such that $d_3 d_4 > 2$), we set $u=\sqrt{d_4 \frac{\log(p)}{n} + t} \le 3$ and apply a union bound to obtain 
	\begin{align*}
	\P{ \pnorm{ \frac{1}{n} \sum_{k=1}^n \B_k }{\infty} \ge \sqrt{d_4 \frac{\log(p)}{n} + t} }
	\le 2p^2 e^{-d_3 d_4 \log(p) - d_3 n t}
	\le 2e^{-d_3 n t}.
	\end{align*}
\end{proof}
We can now complete the proof of Theorem \ref{thm:Operator}.
\begin{proof}[Proof of Theorem \ref{thm:Operator}]
	Recall that $\A = \cos(\arcsin(\SIGMA))$ and 
	$$\GAM = \johannes{\E\round{\sign(\X)\sign(\X)^T}}{} = \frac{2}{\pi} \arcsin(\SIGMA).$$ 
	Let us, for \sjoerd{the sake of clarity}, consider the case $\M = \boldsymbol{1}$. \sjoerd{The general case follows by replacing $\A$ and $\SIGMA$ by $\M\odot\A$ and $\M\odot\SIGMA$ throughout the proof. We will moreover assume throughout that $\lceil\log(p)\rceil\geq \sjoerdgreen{2}$, the remaining case is similar (and easier).}\par 
	\sjoerd{By \eqref{eqn:TaylorExpMat} and the triangle inequality,
		$$\pnorm{\tilde{\SIGMA}_n - \SIGMA}{}\leq \alpha+\beta+\gamma,$$
		where
		\begin{align} 
		\label{eqn:alphabetagammaSplit}
		\alpha & = \pnorm{\A \odot \B - \frac{1}{2} \SIGMA \odot \B^{\odot 2}}{} \nonumber\\
		\beta & = \sum_{\ell=1}^{\lfloor \frac{1}{2} \log(p) \rfloor}\frac{1}{(2\ell+1)!}\pnorm{\A\odot \B^{\odot (2\ell+1)}}{} + \sum_{\ell=1}^{\lfloor \frac{1}{2} \log(p) \rfloor}\frac{1}{(2\ell+2)!}\pnorm{\SIGMA\odot \B^{\odot (2\ell+2)}}{}\\
		\gamma & = \johannes{\sum_{\ell = \lceil \frac{1}{2} \log(p) \rceil}^{\infty}} \frac{1}{(2\ell+1)!}\pnorm{\A\odot \B^{\odot (2\ell+1)}}{} + \johannes{\sum_{\ell = \lceil \frac{1}{2} \log(p) \rceil}^{\infty}} \frac{1}{(2\ell+2)!}\pnorm{\SIGMA\odot \B^{\odot (2\ell+2)}}{}\nonumber
		\end{align}
		We start by estimating $\alpha$. Write $\B = \frac{\pi}{2n} \sum_{k=1}^n \B_k$ and recall from Lemma~\ref{lem:BkHadPowers} that $\B_k^{\odot 2} = \ALPH + \BET \odot \B_k$ with $\ALPH = \ALPH_2 = (\boldsymbol{1} - \GAM^{\odot 2})$ and $\BET = \BET_2 = \sjoerd{-}2\GAM$. This yields
		\begin{align} \label{eqn:decompTaylor1and2_SpecialCase}
		\begin{split}
		&\A\odot \B - \frac{1}{2} (\SIGMA\odot \B^{\odot 2}) \\
		&\quad= \frac{\pi}{2n} \sum_{k=1}^n \A\odot \B_k - \frac{\pi^2}{8n^2} \sum_{k=1}^n \SIGMA \odot \B_k^{\odot 2} - \frac{\pi^2}{8n^2} \sum_{k \neq \ell} \SIGMA\odot \B_k\odot \B_{\ell} \\
		&\quad= \frac{\pi}{2n} \sum_{k=1}^n \A \odot \B_k - \frac{\pi^2}{8n^2} \sum_{k=1}^n \SIGMA \odot\BET \odot \B_k - \frac{\pi^2}{8n} \SIGMA\odot \ALPH - \frac{\pi^2}{8n^2} \sum_{k\neq\ell} \SIGMA\odot \B_k\odot \B_{\ell},
		\end{split}
		\end{align}
		Using Lemma~\ref{lem:decoupling} \johannes{and} Lemma~\ref{lem:allDifferentIndices}, we obtain for all $\log(p)\leq q\leq n$,
		\begin{align*}
		(\mathbb{E}\alpha^q)^{1/q} 
		& \le 
		\frac{\pi^2}{8n} \left\| \SIGMA \odot \ALPH\right\|
		+ \johannes{\frac{\pi}{2n} \round{ \E \left\|\sum_{k=1}^n \A \odot \B_k\right\|^q }^{1/q}} \\
		& \quad \ + \frac{\pi^2}{8n^2} \round{ \E \left \|\sum_{k=1}^n \SIGMA \odot\BET \odot \B_k\right\|^q }^{1/q} \johannes{+ \frac{\pi^2}{8n^2} \round{ \E \left\|\sum_{k\neq\ell} \SIGMA\odot \B_k\odot \B_{\ell}\right\|^q }^{1/q}} \\
		& \quad \lesssim \frac{1}{n} \left\| \SIGMA \odot \ALPH\right\| 
		+ \frac{1}{n}\left(\sqrt{qn}\left\|\sigma\left(\A\right)\right\| 
		+ q\left\| \A \right\|\right) \\
		& \quad \qquad + \frac{1}{n^2}\left(\sqrt{qn}\left\|\sigma\left(\SIGMA \odot \BET \right)\right\| 
		+ q\left\|\SIGMA \odot \BET \right\|\right)  \\
		&\quad \qquad \qquad + \frac{1}{n^2}(\sqrt{qn}+q)(\sqrt{qn}\|\sigma(\SIGMA)\| + q\|\SIGMA\|) \\
		&\quad \lesssim 
		\sqrt{\frac{q}{n}}\|\sigma(\sjoerdgreen{\A})\|+\frac{q}{n}\max\curly{ \pnorm{\A}{}, \pnorm{\SIGMA}{} }
		\end{align*}
		where we used in the last inequality that $q\leq n$,  $\|\sigma(\SIGMA)\|^2 \le 2 \pnorm{\SIGMA}{}^2$ and $\|\sigma(\SIGMA\odot\BET)\|^2 \le 2 \pnorm{\SIGMA\odot\BET}{}^2$ by \sjoerd{Lemma~\ref{lem:sigmavsOperatornorm}} and that $\pnorm{\SIGMA \odot \ALPH}{},\pnorm{\SIGMA \odot \BET}{} \le 4 \pnorm{\SIGMA}{}$ by Lemma \ref{lem:BkHadPowers}.\par
		Let us now estimate $\beta$. By the triangle inequality and Lemma~\ref{lem:HadPowersLessThanLogp}, \sjoerdgreen{we obtain for all $\log(p) \le q \le \frac{n}{\log^2(p)}$}
		\begin{align}
		\label{eqn:betaLqEstimate}
		\begin{split}
		    (\mathbb{E}\beta^q)^{1/q} 
		    & \leq \sum_{\ell=1}^{\lfloor \frac{1}{2} \log(p) \rfloor}\frac{1}{(2\ell+1)!} \johannes{\round{ \mathbb{E}\pnorm{\A\odot \B^{\odot (2\ell+1)}}{}^q}^{1/q}} \\
		    &\qquad + \sum_{\ell=1}^{\lfloor \frac{1}{2} \log(p) \rfloor}\frac{1}{(2\ell+2)!} \johannes{\round{ \mathbb{E}\pnorm{\SIGMA\odot \B^{\odot (2\ell+2)}}{}^q}^{1/q}}{} \\
		    & \leq \sum_{\ell=1}^{\lfloor \frac{1}{2} \log(p) \rfloor} \left(\frac{c\sjoerdgreen{q}}{n}\right)^{(2\ell+1)/2}\|\A\| + \sum_{\ell=1}^{\lfloor \frac{1}{2} \log(p) \rfloor}\left(\frac{c\sjoerdgreen{q}}{n}\right)^{(2\ell+2)/2}\|\SIGMA\|\\
		    & \leq \left(\frac{c\sjoerdgreen{q}}{n}\right)^{3/2}\|\A\| \sum_{\ell=1}^{\infty} \left(\frac{c\sjoerdgreen{q}}{n}\right)^{\ell-1} + \left(\frac{c\sjoerdgreen{q}}{n}\right)^{2}\|\SIGMA\| \sum_{\ell=1}^{\infty}\left(\frac{c\sjoerdgreen{q}}{n}\right)^{\ell-1}
		\end{split}
		\end{align}
		\sjoerdgreen{and hence} 
		$$(\mathbb{E}\beta^q)^{1/q} \quad \lesssim \frac{q}{n}\max\{\|\A\|,\|\SIGMA\|\}.$$
		In summary, if $\log(p)\leq q\leq \frac{n}{\sjoerdgreen{c\log^2(p)}}$, then 
		$$(\mathbb{E}(\alpha+\beta)^q)^{1/q} \sjoerdgreen{\lesssim} \sqrt{\frac{q}{n}} \|\sigma(\A)\|
		+\frac{q}{n}\max\curly{ \pnorm{\A}{}, \pnorm{\SIGMA}{} }.$$
		Lemma~\ref{lem:LqtoTailBound} yields, for all $n\geq \sjoerdgreen{c\log^2(p)}(\log(p)+u)$,
		\begin{align*}
		& \johannes{\P{ \alpha+\beta\geq 2e\sjoerdgreen{C}\left(\sqrt{\frac{\log(p)+u}{n}} \|\sigma(\A)\| + \frac{\log(p)+u}{n}\max\curly{ \pnorm{\A}{}, \pnorm{\SIGMA}{} }\right) }}{}  \\
		& \qquad \qquad \qquad \qquad \qquad \qquad \qquad \qquad \qquad\qquad \qquad \qquad \qquad\leq e^{-\sjoerdgreen{(u+\log(p))}}.
		\end{align*}
		It remains to estimate $\gamma$. Using Lemma~\ref{lem:HadamardOperatorNorm} and $\pnorm{\Z}{1\rightarrow 2} \leq \sqrt{p} \pnorm{\Z}{\infty}$ for $\Z\in \R^{p\times p}$, we find 
		\begin{align}
		\label{eqn:gammaDeterministicEstimate}
        \gamma & = \sum_{\ell=\lceil \frac{1}{2} \log(p) \rceil}^{\infty}\frac{1}{(2\ell+1)!}\pnorm{\A\odot \B^{\odot (2\ell+1)}}{} + \sum_{\ell=\lceil \frac{1}{2} \log(p) \rceil}^{\infty}\frac{1}{(2\ell+2)!}\pnorm{\SIGMA\odot \B^{\odot (2\ell+2)}}{} \notag\\
        &\leq\max\{\|\A\|,\|\SIGMA\|\}\sum_{k=2\lceil \frac{1}{2} \log(p) \rceil}^{\infty} \frac{1}{k!} \pnorm{ \B^{\odot k} }{1 \rightarrow 2} \notag \\
        & \leq\max\{\|\A\|,\|\SIGMA\|\}\sum_{k=\lceil \log(p) \rceil}^{\infty} \frac{\sqrt{p}}{k!} \pnorm{ \B}{\infty}^k \leq\max\{\|\A\|,\|\SIGMA\|\}\sum_{k=\lceil \log(p) \rceil}^{\infty} \frac{e^{k/2}}{k!} \pnorm{ \B}{\infty}^k,  
		\end{align}
		as $\sqrt{p} \leq e^{k/2}$ for $k \ge \log(p)$. Recall that with probability at least $1 - 2e^{-c_2 u }$
		\begin{align*} 
		\pnorm{\B}{\infty} \le \sqrt{ c_1 \frac{\log(p) + u}{n}},
		\end{align*}
		see Lemma~\ref{lem:MaxBound}. Upon this event, 
		\begin{align*}
		\gamma &\leq\max\{\|\A\|,\|\SIGMA\|\}\sum_{k=\lceil \log(p) \rceil}^{\infty} \frac{1}{k!} \left(ec_1 \frac{\log(p) + u}{n}\right)^{k/2} \lesssim \max\{\|\A\|,\|\SIGMA\|\}\frac{\log(p) + u}{n},
		\end{align*}
		since $\lceil\log(p)\rceil\geq \sjoerdgreen{2}$ and $n\gtrsim \log(p)+u$. Combining the estimates for $\alpha+\beta$ and $\gamma$ yields the assertion.}
\end{proof}

\section{Proof of Proposition~\ref{prop:LB}}
\label{sec:proofNoDithLower}

Using similar techniques as in Section \ref{sec:proofNoDith}, we can derive \ae{the lower bound in Proposition \ref{prop:LB}. It} agrees with the upper bound in Theorem~\ref{thm:Operator} in its dominant term.\par
\sjoerd{In the proof we will use the following special case of \cite[Proposition 5.29]{dirksen2011noncommutative}. We include a short proof for the sake of completeness. 
	\begin{lemma}
		\label{lem:triangleIneqSigma}
		Let $d_1,d_2\in \mathbb{N}$. The map $\W\mapsto\|(\mathbb{E}(\W^T \W))^{1/2}\|$ defines a norm on \sjoerdgreen{the space of random matrices in} $\R^{d_1\times d_2}$. In particular, the map $\Z\mapsto \|\sigma(\Z)\|$ on the symmetric matrices in $\R^{p\times p}$ satisfies the (reverse) triangle inequality.     
	\end{lemma}
	\begin{proof}
		\johannes{We only prove the triangle inequality. The remaining norm properties are immediate.} For any $\W,\Z\in\R^{d_1\times d_2}$ and $t>0$, $(t\W-\frac{1}{t}\Z)^T (t\W-\frac{1}{t}\Z)\succeq 0$ and hence
		$$(\W+\Z)^T(\W+\Z) \preceq (1+t^2)\W^T\W+(1+t^{-2}) \Z^T\Z.$$
		Taking expectations and norms and using the triangle inequality, we obtain
		$$\|\mathbb{E}[(\W+\Z)^T(\W+\Z)]\|\leq (1+t^2)\|\mathbb{E}(\W^T\W)\|+(1+t^{-2})\|\mathbb{E}(\Z^T\Z)\|.$$
		Minimizing the right hand side, we find
		$$\|\mathbb{E}[(\W+\Z)^T(\W+\Z)]\|\leq (\|\mathbb{E}(\W^T\W)\|^{1/2}+\|\mathbb{E}(\Z^T\Z)\|^{1/2})^2,$$
		which shows that the triangle inequality holds.\par
		The final statement is an immediate consequence, as
		$$\|\sigma(\Z)\| = \|(\mathbb{E}[(\Z\odot\B_1)^T(\Z\odot\B_1)])^{1/2}\|,$$
		with $\B_1$ defined as in \eqref{def:B}.
	\end{proof}
	We are now ready to prove the lower bound.} 

\begin{proof}[Proof of Proposition \ref{prop:LB}]
	\sjoerd{As in the proof of Theorem~\ref{thm:Operator} we restrict ourselves for convenience to the case where $\M = \boldsymbol{1}$ and $\lceil\log(p)\rceil\geq \sjoerdgreen{3}$. The adaption to the general case is straightforward. By the (reverse) triangle inequality, we can write 
		\begin{equation}
		\label{eqn:reverseTIsplitalphabetagamma}
		\johannes{ \round{ \E\pnorm{\tilde{\SIGMA}_n - \SIGMA}{}^2}^{\frac{1}{2}} }\geq (\mathbb{E}\alpha^2)^{1/2}-(\E\beta^2)^{1/2}-(\E\gamma^2)^{1/2},\end{equation}
		with $\alpha,\beta,\gamma$ as defined in \eqref{eqn:alphabetagammaSplit}. We start by estimating $(\mathbb{E}\alpha^2)^{1/2}$. Taking expectations in \eqref{eqn:decompTaylor1and2_SpecialCase},} we find
	\begin{align*}
	&\mathbb{E} \round{ \A \odot \B - \frac{1}{2} \SIGMA \odot \B^{\odot 2} }^2  \\
	& \quad \sjoerd{= \mathbb{E} \round{ \frac{\pi}{2n} \sum_{k=1}^n \left(\A - \frac{\pi}{4n} \SIGMA\odot\BET \right)\odot \B_k - \frac{\pi^2}{8n} \SIGMA\odot \ALPH  - \frac{\pi^2}{8n^2} \sum_{k\neq\ell} \SIGMA\odot \B_k\odot \B_{\ell} }^2}\\
	&\quad = \mathbb{E} \round{ \frac{\pi}{2n} \sum_{k=1}^n \left(\A - \frac{\pi}{4n} \SIGMA\odot\BET \right)\odot \B_k }^2  
	+ \round{ \frac{\pi^2}{8n} \SIGMA\odot \ALPH }^2 \\
	& \quad \qquad \qquad \qquad \qquad 
	+ \mathbb{E} \round{ \frac{\pi^2}{8n^2} \sum_{k\neq\ell} \SIGMA\odot \B_k\odot \B_{\ell} }^2,
	\end{align*}
	\sjoerd{where we use that the expectation of all ``cross-terms'' in the expansion of the square are zero as $\B_k$ and $\B_{\ell}$ are independent for all $k\neq \ell$ and $\E \B_k  = \0$ for all $k \in [n]$.} Since all the terms are positive semidefinite,
	\begin{align*}
    \mathbb{E}\alpha^2 & \geq \left\|\mathbb{E}\round{ \A \odot \B - \frac{1}{2} \SIGMA \odot \B^{\odot 2} }^2 \right\| \\
	& \simeq \max\left\{\left\|\mathbb{E} \round{ \frac{\pi}{2n} \sum_{k=1}^n \left(\A - \frac{\pi}{4n} \SIGMA\odot\BET \right)\odot \B_k }^2 \right\|,
	\right. \\
	& \qquad \qquad \qquad \qquad \qquad \left. \left\|\round{ \frac{\pi^2}{8n} \SIGMA\odot \ALPH }^2\right\|, 
	\left\|\mathbb{E}\round{ \frac{\pi^2}{8n^2} \sum_{k\neq\ell} \SIGMA\odot \B_k\odot \B_{\ell} }^2 \right\|\right\}.
	\end{align*}
	Recall that $\B_{\sjoerd{k}} =  \Y^k (\Y^k)^T - \johannes{\E\round{\Y^k (\Y^k)^T}} = \Y^k (\Y^k)^T - \GAM$. Since
	\begin{align} \label{eq:Rank1Trick}
	\C \odot \Y\Y^T = \diag(\Y) \cdot \C \cdot \diag(\Y)
	\end{align} 
	for all $\C \in \R^{p\times p}$, and $\diag(\Y)^2 = \id$, we find
	\begin{align*}
	& \mathbb{E} \round{ \frac{\pi}{2n} \sum_{k=1}^n \left(\A - \frac{\pi}{4n} \SIGMA\odot\BET \right)\odot \B_k }^2 \\ 
	&\qquad = \frac{\pi^2}{4n} \mathbb{E}\round{ \left(\A - \frac{\pi}{4n} \SIGMA\odot\BET \right) \odot \B_1 }^2  \\
	&\qquad= \frac{\pi^2}{4n} \round{  \left(\A - \frac{\pi}{4n} \SIGMA\odot\BET \right)^2 \odot \GAM - \round{ \left(\A - \frac{\pi}{4n} \SIGMA\odot\BET \right) \odot \GAM }^2  } \\
	&\qquad = \frac{\pi^2}{4n} \sigma \round{\A - \frac{\pi}{4n} \SIGMA\odot\BET}^2.
	\end{align*}
	\sjoerd{Using Lemma~\ref{lem:triangleIneqSigma} we obtain
		\begin{align*}
		& \left\|\left(\mathbb{E} \round{ \frac{\pi}{2n} \sum_{k=1}^n \left(\A - \frac{\pi}{4n} \SIGMA\odot\BET \right)\odot \B_k }^2\right)^{1/2}\right\| \\
		& \qquad = \frac{\pi}{\sqrt{n}} \left\|\sigma \round{\A - \frac{\pi}{4n} \SIGMA\odot\BET}\right\|\\
		& \qquad \geq \frac{\pi}{\sqrt{n}}\|\sigma(\A)\| - \frac{\pi^2}{4n^{3/2}}\|\sigma(\SIGMA\odot\BET)\| \sjoerdgreen{\geq \frac{\pi}{\sqrt{n}}\|\sigma(\A)\| - \frac{\pi^2}{2n^{3/2}}\|\SIGMA\|.}
		\end{align*}
	}Moreover, by independence of $\B_k$ and $\B_\ell$, for $k \neq \ell$, and $\johannes{\mathbb{E} \B_k}{}  = \0$, for all $k \in [n]$, 
	\begin{align*}
	& \mathbb{E} \round{ \frac{\pi^2}{8n^2} \sum_{k\neq\ell} \SIGMA\odot \B_k\odot \B_{\ell} }^2 \\
	& \qquad = \frac{\pi^4}{64 n^4} \sum_{k_1 \neq \ell_1} \sum_{k_2 \neq \ell_2} \johannes{ \E \round{\round{ \SIGMA \odot \B_{k_1} \odot \B_{\ell_1} } \round{ \SIGMA \odot \B_{k_2} \odot \B_{\ell_2} } }} \\
	& \qquad = \frac{\pi^4}{32 n^4} \sum_{k \neq \ell} \johannes{ \E \round{ \round{ \SIGMA \odot \B_{k} \odot \B_{\ell} } \round{ \SIGMA \odot \B_{k} \odot \B_{\ell} } } } \\
	& \qquad = \frac{\pi^4 n (n-1)}{64 n^4} \E{ \round{ \SIGMA \odot \B_{1} \odot \B_{2} } \round{ \SIGMA \odot \B_{1} \odot \B_{2} } } \\
	& \qquad = \frac{\pi^4 n (n-1)}{64 n^4} \round{\SIGMA^2 \odot \GAM^2 - (\SIGMA \odot \GAM)^2 \odot \GAM} = \frac{\pi^4 n (n-1)}{64 n^4} \sigma (\SIGMA)^2 \odot \GAM.
	\end{align*}
	\sjoerd{In the penultimate step we used that \eqref{eq:Rank1Trick} implies}
	\begin{align*}
	\johannes{ \E\round{ (\C \odot \B_k) (\D \odot \B_k) } }
	= \johannes{ \E \round{ \C\D } }  \odot \GAM - ( \johannes{\mathbb{E}\C}{}\odot \GAM) (\johannes{\mathbb{E}\D}{} \odot \GAM)
	\end{align*}
	\sjoerd{for any random matrices $\C,\D \in \R^{p\times p}$ that are independent of $\B_k$, so in particular }
	\begin{align*}
	\johannes{\E\round{ \round{ \SIGMA \odot \B_{1} \odot \B_{2} } \round{ \SIGMA \odot \B_{1} \odot \B_{2} } } }{}
	&= \E{ (\SIGMA \odot \B_1)^2 } \odot \GAM - (\johannes{\mathbb{E} (\SIGMA \odot \B_1)}{} \odot \GAM)^2 \\
	&= \SIGMA^2 \odot \GAM^{\odot 2} - (\SIGMA \odot \GAM)^2 \odot \GAM.
	\end{align*}
	\sjoerd{To estimate the remaining terms in \eqref{eqn:reverseTIsplitalphabetagamma}, note that 
		$$(\mathbb{E}\beta^2)^{1/2}\leq (\mathbb{E}\beta^{\log(p)})^{1/\log(p)} \leq \max\{\|\A\|,\|\SIGMA\|\} \left(\frac{c\sjoerdgreen{\log(p)}}{n}\right)^{3/2}$$
		if $n \sjoerdgreen{ \gtrsim\log^3(p)}$ by \eqref{eqn:betaLqEstimate}. Moreover, by \eqref{eqn:gammaDeterministicEstimate}, 
		$$(\mathbb{E}\gamma^2)^{1/2}\leq \max\{\|\A\|,\|\SIGMA\|\}\sum_{k=\lceil \log(p) \rceil}^{\infty} \frac{e^{k/2}}{k!} (\mathbb{E}\pnorm{ \B}{\infty}^{2k})^{1/2}.$$
		Using that $k\geq \log(p)$ and applying Bernstein's inequality, 
		\begin{align*}
		(\mathbb{E}\pnorm{ \B}{\infty}^{2k})^{1/(2k)} & = \frac{\pi}{2}\left(\mathbb{E}\max_{1\leq i,j\leq p}\left|\frac{1}{n}\sum_{m=1}^n(Y_{i,j}^m-\mathbb{E}[Y_{i,j}^m])\right|^{2k}\right)^{1/(2k)}\\
		&\leq \frac{e\pi}{2}\max_{1\leq i,j\leq p} \left(\mathbb{E}\left|\frac{1}{n}\sum_{m=1}^n(Y_{i,j}^m-\mathbb{E}[Y_{i,j}^m])\right|^{2k}\right)^{1/(2k)} \lesssim \sqrt{\frac{k}{n}}+\frac{k}{n}.
		\end{align*}
		Using this estimate together with $k!\geq k^k e^{-k}$, we obtain
		\begin{align*}
		(\mathbb{E}\gamma^2)^{1/2}& \leq \max\{\|\A\|,\|\SIGMA\|\}\sum_{k=\lceil \log(p) \rceil}^{\infty} \frac{c_1^{k}}{k!} \left(\frac{k^{k/2}}{n^{k/2}} + \frac{k^k}{n^k}\right) \\
		& \leq \max\{\|\A\|,\|\SIGMA\|\}\sum_{k=\lceil \log(p) \rceil}^{\infty} \left(\frac{c_2}{n}\right)^{k/2} \lesssim \max\{\|\A\|,\|\SIGMA\|\} \left(\frac{c_2}{n}\right)^{3/2},
		\end{align*}
		as $\lceil \log(p) \rceil\geq \sjoerdgreen{3}$. This completes the proof.}  
\end{proof}

\section{Proof of Theorem~\ref{thm:OperatorDitheredMask}}
\label{sec:Dithered}

To prove Theorem~\ref{thm:OperatorDitheredMask}, we first estimate the bias of $\tilde{\SIGMA}'_n$, see  Lemma~\ref{lem:linftyBiasEst}. Its proof relies on Lemmas~\ref{lem:Expectation} and \ref{lem:biasSignProd}. The remainder of the proof is a relatively straightforward application of the matrix Bernstein inequality (Theorem~\ref{thm:sumsRMs}). 
\begin{lemma} \label{lem:Expectation}
	Let $a,b\in \R$, let $\lambda>\max\{|a|,|b|\}$ and let $\sigma,\sigma'$ be independent and uniformly distributed in $[-\lambda,\lambda]$. Then, 
	$$\johannes{\E (\sign(a+\sigma)\sign(b+\sigma'))}{} = \frac{ab}{\lambda^2}.$$
\end{lemma}{}
\begin{proof}    
	By independence of $\sigma$ and $\sigma'$,
	\begin{align*}
	&\johannes{\E\round{\sign(a + \sigma) \sign(b + \sigma')}}{} \\
	&= \P{\sign(a + \sigma) = \sign(b + \sigma')} - \P{\sign(a + \sigma) \neq \sign(b + \sigma')} \\
	&= 2\P{\sign(a + \sigma) = \sign(b + \sigma')} - 1 \\
	&= 2\left( \P{\sigma > -a, \; \sigma' > -b} + \P{\sigma < -a, \; \sigma' < -b} \right) - 1 \\
	&= 2\left( \frac{1}{4\lambda^2} (\lambda + a)(\lambda + b) + \frac{1}{4\lambda^2} (\lambda - a)(\lambda - b) \right) - 1 \\
	&= \frac{1}{2\lambda^2} \left( 2\lambda^2 + 2 a b \right) - 1 = \frac{a b}{\lambda^2}.
	\end{align*}{}
	This completes the proof.
\end{proof}
In the proof of the following lemma we use that for any subgaussian random variable $W$ and any $\mu>0$,
\begin{align}
\label{eqn:truncEstSubg}
\johannes{\E\round{|W|1_{\{|W|>\mu\}}}}{} & = \int_0^{\infty} \P{|W|1_{\{|W|>\mu\}}>t} \ dt 
= \mu \P{|W|>\mu} + \int_{\mu}^{\infty} \P{|W|>t} \ dt \nonumber \\
& \lesssim \mu e^{-c\mu^2/\|W\|_{\psi_2}^2} + \frac{\|W\|_{\psi_2}^2}{\mu}e^{-c\mu^2/\|W\|_{\psi_2}^2},
\end{align}
\sjoerd{where the final inequality follows from the well-known estimate
	$$\int_u^{\infty} e^{-t^2/2} \ dt \leq \frac{1}{u}e^{-u^2/2}, \qquad \sjoerdgreen{u>0}.$$
	Similarly,} for any subexponential random variable $Z$, 
\begin{align}
\label{eqn:truncEstSubexp}
\johannes{\E\round{ |Z|1_{\{|Z|>\mu\}} } }{}
\lesssim \mu e^{-c\mu/\|Z\|_{\psi_1}} + \|Z\|_{\psi_1}e^{-c\mu/\|Z\|_{\psi_1}}.
\end{align}
\begin{lemma} \label{lem:biasSignProd}
	Let $U,V$ be subgaussian random variables, let $\lambda>0$, and \sjoerd{let} $\sigma,\sigma'$ be independent, uniformly distributed in $[-\lambda,\lambda]$ and independent of $U$ and $V$. Then, 
	$$\left| \johannes{\E\round{\lambda^2\sign(U + \sigma) \sign(V + \sigma')} - \E\round{UV}}{} \right|
	\lesssim
	(\lambda^2+\theta_{U,V}^2)e^{-c\lambda^2/\theta_{U,V}^2},$$
	where $\theta_{U,V} = \max\{\|U\|_{\psi_2},\|V\|_{\psi_2}\}$.
\end{lemma}
\begin{proof}
	Since $U$ and $V$ are independent of $\sigma$ and $\sigma'$, Lemma~\ref{lem:Expectation} yields
	$$\johannes{\E_{\sigma,\sigma'}\round{\lambda^2\sign(U + \sigma) \sign(V + \sigma')}}{} 1_{\{|U|\leq \lambda, |V|\leq \lambda\}} = UV 1_{\{|U|\leq \lambda, |V|\leq \lambda\}}.$$
	Hence,
	\begin{align} \label{eq:Claim}
	&\johannes{\E\round{\lambda^2\sign(U + \sigma) \sign(V + \sigma')} - \E(UV)}{} \\
	&\quad = \johannes{\E\round{(\lambda^2\sign(U + \sigma)\sign(V + \sigma')-UV)1_{\{|U|>\lambda\}\cup\{|V|>\lambda\}}}}{}. \notag
	\end{align}
	Since $U,V$ are subgaussian, 
	\begin{align} \label{eq:Part1}
	\left| \johannes{\E\round{\lambda^2\sign(U + \sigma)\sign(V + \sigma')1_{\{|U|>\lambda\}\cup\{|V|>\lambda\}}}}{} \right| 
	&\leq 
	\lambda^2(\mathbb{P}[|U|>\lambda]+\mathbb{P}[|V|>\lambda] ) \\
	&\leq 2\lambda^2 e^{-c\lambda^2/\theta_{U,V}^2}. \notag
	\end{align}
	Moreover, 
	\begin{align}
	\left| \johannes{\E\round{UV1_{\{|U|>\lambda\}\cup\{|V|>\lambda\}}}}{} \right|
	\leq 
	\johannes{\E\round{|UV| 1_{\{|U|>\lambda\}}} + \E\round{|UV| 1_{\{|V|>\lambda\}}}}{}.
	\end{align}
	By \eqref{eqn:truncEstSubg} and \eqref{eqn:truncEstSubexp}
	\begin{align} 
	& \johannes{\E\round{|UV|1_{\{|U|>\lambda\}}} }{}\notag\\
	&\qquad = \johannes{\E\round{|UV| (1_{\{|U|>\lambda, |V|>\lambda\}} + 1_{\{|U|>\lambda, |V|\leq\lambda\}})}}{} \notag \\
	&\qquad \leq \johannes{\E\round{|UV|1_{\{|UV|>\lambda^2\}}} + \lambda \E\round{|U|1_{\{|U|>\lambda\}}}}{}\notag\\
	&\qquad \leq (\lambda^2+\|UV\|_{\psi_1})e^{-c\lambda^2/\|UV\|_{\psi_1}} + \lambda\Big(\lambda e^{-c\lambda^2/\|U\|_{\psi_2}^2} + \frac{\|U\|_{\psi_2}^2}{\lambda}e^{-c\lambda^2/\|U\|_{\psi_2}^2}\Big) \notag\\
	&\qquad \leq 2(\lambda^2+\theta_{U,V}^2)e^{-c\lambda^2/\theta_{U,V}^2},\label{eq:Part2}
	\end{align}
	where we have used $\|UV\|_{\psi_1}\leq \|U\|_{\psi_2}\|V\|_{\psi_2}$. The claim follows \sjoerd{by} using \eqref{eq:Part1}-\eqref{eq:Part2} in \eqref{eq:Claim}.
\end{proof}
\begin{lemma}
	\label{lem:linftyBiasEst}
	There exists constant\sjoerd{s $c_1,c_2>0$} depending only on $K$ such that the following holds. Let $\X$ be a mean-zero, $K$-subgaussian vector with covariance matrix $\johannes{\E\round{ \X\X^T }}{} = \SIGMA$. Let \sjoerd{$\lambda>0$ and let} $\Y=\sign(\X + \Tau)$ and $\bar{\Y}=\sign(\X + \bar{\Tau})$, where $\Tau,\bar{\Tau}$ are independent and uniformly distributed in $[-\lambda,\lambda]^p$ \sjoerd{and independent of $\X$}. Then,
	$$\pnorm{\lambda^2\johannes{\E\round{\Y\bar{\Y}^T}}{} - \SIGMA}{\infty}\sjoerd{\leq c_1} (\lambda^2+\|\SIGMA\|_{\infty})e^{-\sjoerd{c_2}\lambda^2/\|\SIGMA\|_{\infty}}$$
	and 
	$$\pnorm{\lambda^2\johannes{\E\round{\Y\Y^T}}{} - (\SIGMA-\diag(\SIGMA)+\lambda^2 \id)}{\infty}\sjoerd{\leq c_1} (\lambda^2+\|\SIGMA\|_{\infty})e^{-\sjoerd{c_2}\lambda^2/\|\SIGMA\|_{\infty}}.$$
\end{lemma}
\begin{proof}
	Since $\X$ is $K$-subgaussian, for any $\ell\in[p]$,
	$$\|X_{\ell}\|_{\psi_2} = \|\langle \X,\mathbf{e}_{\ell}\rangle\|_{\psi_2} \leq K \johannes{(\mathbb{E}\langle \X,\mathbf{e}_{\ell}\rangle^2)^{1/2}}{} = K\SIGMA_{\ell\ell}^{1/2}\leq K\|\SIGMA\|_{\infty}^{1/2}.$$
	Lemma~\ref{lem:biasSignProd} applied for $U=X_i$ and $V=X_j$ yields 
	$$\abs{ \johannes{\E\round{ \lambda^2 \sign(X_i + \tau_i) \sign(X_j + \bar\tau_j)}}{} - \Sigma_{i,j} }{} \lesssim (\lambda^2+\sjoerdgreen{K^2}\|\SIGMA\|_{\infty})e^{-c_1\lambda^2/\sjoerdgreen{K^2}\|\SIGMA\|_{\infty}}$$ 
	for all $i,j\in [p]$ and  
	$$\abs{ \johannes{\E\round{ \lambda^2 \sign(X_i + \tau_i) \sign(X_j + \tau_j)}}{} - \Sigma_{i,j} }{} \lesssim (\lambda^2+\sjoerdgreen{K^2}\|\SIGMA\|_{\infty})e^{-c_1\lambda^2/\sjoerdgreen{K^2}\|\SIGMA\|_{\infty}}$$
	whenever $i\neq j$. These two observations immediately imply the two statements. 
\end{proof}
We are now ready to prove the main result of this section.
\begin{proof}[Proof of Theorem~\ref{thm:OperatorDitheredMask}]
	Recall the definitions of $\SIGMADITH_n$ and $\tilde{\SIGMA}'_n$ in \eqref{eq:TwoBitEstimator} and \eqref{eq:AsymmetricEstimator}. Clearly,
	$$\pnorm{\M \odot \SIGMADITH_n - \M \odot \SIGMA}{} \leq \pnorm{\M \odot \tilde{\SIGMA}'_n - \M \odot \SIGMA}{}$$
	and
	\begin{align}
	\label{eqn:splitExpDithered}
	\pnorm{\M \odot \tilde{\SIGMA}'_n - \M \odot \SIGMA}{}
	\le \pnorm{ \M \odot \tilde{\SIGMA}'_n - \M \odot \E{\tilde{\SIGMA}'_n}}{} + \pnorm{ \M \odot \round{ \E{\tilde{\SIGMA}'_n} - \SIGMA }}{}.
	\end{align}
	By noting that $\M$ has only \sjoerd{nonnegative} entries and applying Lemma~\ref{lem:linftyBiasEst}, we obtain 
	\begin{align}
	\label{eqn:normEstPos}
	\pnorm{ \M \odot \round{ \E{\tilde{\SIGMA}'_n} - \SIGMA }}{} \nonumber
	&=  \sup_{\v \in \R^p \ : \ \|\v\|_2\leq 1} \abs{\sum_{i,j=1}^p M_{i,j} \round{ \E{\tilde{\SIGMA}'_n} - \SIGMA}_{i,j} v_i v_j }{} \nonumber\\
	&
	\le  \sup_{\v \in \R^p\ : \ \|\v\|_2\leq 1} \sum_{i,j=1}^p M_{i,j} \abs{ \round{ \E{\tilde{\SIGMA}'_n} - \SIGMA}_{i,j} }{} \abs{v_i v_j}{}  \nonumber\\
	&\le \pnorm{\M}{} \pnorm{\E{\tilde{\SIGMA}'_n} - \SIGMA}{\infty}
	= \pnorm{\M}{} \pnorm{\lambda^2 \johannes{\E\round{\Y\bar{\Y}^T}}{} - \SIGMA}{\infty}
	\nonumber \\
	&
	\sjoerd{\lesssim_K} \pnorm{\M}{} (\lambda^2+\|\SIGMA\|_{\infty})e^{-c_1\lambda^2/\|\SIGMA\|_{\infty}}
	\sjoerd{\lesssim_K} \frac{\lambda^2\|\M\|}{n}
	\end{align}
	provided that $\lambda^2\geq \frac{1}{c_1} \|\SIGMA\|_{\infty}\log(n)$. Let us write $\Y_k = \sign(\X^k + \Tau^k)$ and $\bar{\Y}_k = \sign(\X^k + \bar{\Tau}^k)$. We estimate the first term in \eqref{eqn:splitExpDithered} by defining random matrices 
	\begin{align*}
	{\XI}_k 
	= \frac{\lambda^2}{n} \M\odot \round{ \Y_k\bar{\Y}_k^T - \johannes{\E\round{\Y_k\bar{\Y}_k^T} } }{}
	\quad \text{ \sjoerd{so} that } \quad
	\M \odot \tilde{\SIGMA}'_n - \M \odot \E{\tilde{\SIGMA}'_n} = \sum_{k=1}^n \XI_k.
	\end{align*}
	For any $1\leq k\leq n$,
	\begin{equation}
	\label{eqn:XikMaxNormEstimate}
	\|{\XI}_k\| = \frac{\lambda^2}{n} \|\M\odot(\Y_k\bar{\Y}_k^T - \johannes{\E\round{\Y_k\bar{\Y}_k^T}}{})\| \leq \frac{\sjoerd{2}\lambda^2\|\M\|}{n}
	\end{equation}
	since $\M \odot \Y_k \bar{\Y}_k^T = \diag(\Y_k) \; \M  \; \diag(\bar{\Y}_k)$ and $\pnorm{\diag(\Y_k)}{} = 1$.
	Moreover, using that $(\Y,\bar{\Y})$ and $(\bar{\Y},\Y)$ are identically distributed, we get 
	\begin{align*}
	& \Big\|\Big(\sum_{k=1}^n \sjoerd{\mathbb{E}(\XI_k^T\XI_k)}\Big)^{1/2}\Big\| \\
	&\qquad = \frac{\lambda^2}{\sqrt{n}} \pnorm{ \mathbb{E}[(\M \odot \Y\bar{\Y}^T)^T(\M \odot \Y\bar{\Y}^T)] - \Big(\M \odot \johannes{\E\round{ \Y\bar{\Y}^T }}{} \Big)^T \Big(\M \odot \johannes{\E\round{ \Y\bar{\Y}^T }}{} \Big) }{}^{1/2}\\
	&\qquad= \frac{\lambda^2}{\sqrt{n}} \pnorm{ \M^2\odot\johannes{\E(\bar{\Y}\bar{\Y}^T)}{} - \Big( \M \odot \johannes{\E\round{ \Y\bar{\Y}^T }}{} \Big)^T \Big( \M \odot \johannes{\E\round{ \Y\bar{\Y}^T }}{} \Big) }{}^{1/2}\\
	&\qquad= \frac{\lambda^2}{\sqrt{n}}  \pnorm{ \M^2\odot \johannes{\E(\Y\Y^T)}{} - \Big( \M \odot \johannes{\E\round{ \bar{\Y}\Y^T }}{} \Big)^T \Big( \M \odot \johannes{\E\round{ \bar{\Y}\Y^T }}{} \Big) }{}^{1/2}
	\end{align*}
	as, using $\diag(\Y)^2 = \id$,
	$$ (\M\odot \Y\bar{\Y}^T)^T (\M\odot \Y\bar{\Y}^T)
	= \diag(\bar{\Y})\;\M\; \diag(\Y)\;\diag(\Y)\;\M\;\diag(\bar{\Y})
	=\M^2 \odot \bar{\Y}\bar{\Y}^T.$$
	Interchanging the roles of $\Y_k$ and $\bar{\Y}_k$ yields
	$$\Big\|\Big(\sum_{k=1}^n \sjoerd{\mathbb{E}(\XI_k\XI_k^T)} \Big)^{1/2}\Big\|  = \frac{\lambda^2}{\sqrt{n}} \pnorm{ \M^2\odot \johannes{\E(\Y\Y^T)}{} - (\M\odot \johannes{\E(\bar{\Y}\Y^T)}{})^T (\M\odot \johannes{\E(\bar{\Y}\Y^T)}{}) }{}^{1/2}.$$
	Since 
	$$(\M\odot \johannes{\E(\bar{\Y}\Y^T)}{})^T (\M\odot \johannes{\E(\bar{\Y}\Y^T)}{}) \preceq
	\E{(\M\odot\bar{\Y}\Y^T)^T(\M\odot\bar{\Y}\Y^T)}$$ 
	by Kadison's inequality \eqref{eq:Kadison}, we find
	\begin{align*}
	& \frac{\lambda^2}{\sqrt{n}} \pnorm{ \M^2\odot \johannes{\E(\Y\Y^T)}{} - (\M\odot \johannes{\E(\bar{\Y}\Y^T)}{})^T(\M\odot \johannes{\E(\bar{\Y}\Y^T)}{}) }{}^{1/2}
	\leq  \frac{\sjoerdgreen{2}\lambda^2}{\sqrt{n}} \pnorm{ \M^2\odot \johannes{\E(\Y\Y^T)}{} }{}^{1/2}\\
	& \quad \leq \frac{\sjoerdgreen{2}\lambda}{\sqrt{n}}  \Big(\pnorm{ \M^2\odot (\lambda^2\johannes{\E\round{\Y\Y^T}}{} - (\SIGMA-\diag(\SIGMA)+\lambda^2 \id)) }{} \\
	& \quad \qquad \qquad \qquad \qquad \qquad \qquad \qquad \qquad + \pnorm{ \M^2\odot (\SIGMA-\diag(\SIGMA)+\lambda^2 \id) }{}\Big)^{1/2}
	\end{align*}
	By the same reasoning as in \eqref{eqn:normEstPos} together with Lemma~\ref{lem:linftyBiasEst} 
	\begin{align*}
	&\pnorm{ \M^2 \odot (\lambda^2 \johannes{\E\round{\Y\Y^T}}{} - (\SIGMA-\diag(\SIGMA)+\lambda^2 \id)) }{} \\
	&\quad \leq \|\M^2\| \pnorm{\lambda^2 \johannes{\E\round{\Y\Y^T}}{} - (\SIGMA-\diag(\SIGMA)+\lambda^2 \id)}{\infty} 
	 \sjoerd{\lesssim_K} \|\M\|^2 (\lambda^2+\|\SIGMA\|_{\infty})e^{-c_1\lambda^2/\|\SIGMA\|_{\infty}} 
	\end{align*}
	Hence, 
	\begin{align*}
	&\frac{\lambda^2}{\sqrt{n}} \pnorm{ \M^2\odot \johannes{\E(\Y\Y^T)}{} - (\M\odot \johannes{\E(\bar{\Y}\Y^T)}{})^T(\M\odot \johannes{\E(\bar{\Y}\Y^T)}{}) }{}^{1/2} \\
	&\sjoerd{\lesssim_K} \frac{\lambda}{\sqrt{n}}
	\|\M^2\odot (\SIGMA-\diag(\SIGMA)+\lambda^2 \id) \|^{1/2}  +\frac{\lambda}{\sqrt{n}} \|\M\| (\lambda+\|\SIGMA\|_{\infty}^{1/2})e^{-c_1\lambda^2/2\|\SIGMA\|_{\infty}}.
	\end{align*}
	By Lemma~\ref{lem:HadamardOperatorNorm}, we find
	\begin{align*}
	\|\M^2\odot(\SIGMA-\diag(\SIGMA)+\lambda^2 \id)\| 
	&\leq \|\M^2\odot\SIGMA\|+ \|\M^2\odot \id \odot\SIGMA\| + \lambda^2\|\M^2\odot \id \| \\
	&\leq \|\M\|_{1\to 2}^2(2\|\SIGMA\|+\lambda^2).
	\end{align*}
	In summary, 
	\begin{align}
	\label{eqn:XiksquareNormEstimate}
	& \max \left\{ \left\|\left(\sum_{k=1}^n \sjoerd{\mathbb{E}(\XI_k^T\XI_k)}\right)^{1/2}\right\|, \left\|\left(\sum_{k=1}^n \sjoerd{\mathbb{E}(\XI_k\XI_k^T)}\right)^{1/2}\right\| \right\} \nonumber\\
	& \qquad \sjoerd{\lesssim_K} \frac{1}{\sqrt{n}} \round{ \|\M\|_{1\to 2}(\lambda\|\SIGMA\|^{1/2}+\lambda^2) + \lambda\|\M\| (\lambda+\|\SIGMA\|_{\infty}^{1/2})e^{-c_1\lambda^2/2\|\SIGMA\|_{\infty}} }\nonumber\\
	& \qquad \sjoerd{\lesssim_K} \frac{\|\M\|_{1\to 2}(\lambda\|\SIGMA\|^{1/2}+\lambda^2)}{\sqrt{n}} + \frac{\lambda^2\|\M\|}{n}
	\end{align}
	provided that $\lambda^2\geq \sjoerdgreen{2}\|\SIGMA\|_{\infty}\log(n)/c_1$. By Theorem~\ref{thm:sumsRMs}, \eqref{eqn:XikMaxNormEstimate}, and \eqref{eqn:XiksquareNormEstimate} we find that for any $q\geq 2\log(p)$
	$$\johannes{\round{ \E\left\|\M \odot \tilde{\SIGMA}'_n - \M \odot \bE\left(\tilde{\SIGMA}'_n\right)\right\|^q}^{1/q}}{} \sjoerd{\lesssim_K} \sqrt{q}\frac{\|\M\|_{1\to 2}(\lambda\|\SIGMA\|^{1/2}+\lambda^2)}{\sqrt{n}} + q\frac{\lambda^2\|\M\|}{n}.$$
	\sjoerd{The result now follows immediately from Lemma~\ref{lem:LqtoTailBound}.} 
\end{proof}

\section{Numerical experiments}
\label{sec:Numerics}

\ae{Let us compare our theoretical prediction of the behavior of the two estimators based on quantized samples, $\tilde{\SIGMA}_n$ and $\SIGMADITH_n$ defined in \eqref{eq:OneBitEstimator} and \eqref{eq:TwoBitEstimator}, respectively, to their actual performance in \sjoerd{numerical} experiments. Since we prefer positive semidefinite estimators in practice (see Remark \ref{rem:PSD}), we project both estimators to the set of positive semidefinite matrices. We abuse notation by denoting them by the same symbols. Thus, for the remainder of this section we write
\begin{align} \label{eq:OneBitEstimator2}
\tilde{\SIGMA}_n = P_{\text{PSD}} \round{ \sin\left( \frac{\pi}{2n} \sum_{k=1}^n \sign(\X^k) \sign(\X^k)^T \right) }.
\end{align}
and
\begin{align} \label{eq:TwoBitEstimator2}
\SIGMADITH_n  = P_{\text{PSD}} \round{ \frac{1}{2} \round{\tilde{\SIGMA}'_n + (\tilde{\SIGMA}'_n)^T} },
\end{align}
where $\tilde{\SIGMA}'_n$ was defined in \eqref{eq:AsymmetricEstimator}.}
The \ae{sample covariance matrix} $\hat{\SIGMA}_n$ defined in \eqref{eq:SampleMean}, \sjoerd{which} is computed \sjoerd{using knowledge of the `unquantized' samples,} will serve as a benchmark. \sjoerd{In all our experiments we use i.i.d.\ samples $\X_1,...,\X_n$ from a Gaussian distribution with mean zero and covariance matrix $\SIGMA$. \johannes{ All experiments are averaged over $100$ realizations of the samples $\X_1,...,\X_n$.} \rev{The approximation error is measured in the operator norm.} Moreover, we always tune the dithering parameter $\lambda$ via grid-search on $(0,4\pnorm{\SIGMA}{\infty})$. \johannes{We leave it as an open problem for future research to find a suitable $\lambda$ when \sjoerdgreen{no prior knowledge of $\|\SIGMA\|_{\infty}$ is available}.} }

\subsection{Comparison of all \sjoerd{e}stimators}

In the first experiment, we fix $n = 200$ and, for $p \in [5,30]$, estimate the covariance matrix $\SIGMA \in \R^{p\times p}$ with all proposed estimators from samples $\X_1,...,\X_n$. We construct $\SIGMA$ as having ones on its diagonal and all remaining entries equal to $0.2$. This special form allows fair comparison since it fulfills the more restrictive theoretical assumptions \sjoerd{for the estimator $\tilde{\SIGMA}_n$.}\\
Figure \ref{fig:CompareAll} shows the approximation error of all considered estimators in terms of the operator norm when varying the ambient dimension $p$. Notably, the estimator $\tilde{\SIGMA}_n$, which uses quantized samples without dithering, \sjoerdgreen{performs} almost \rev{as well} as the sample covariance matrix $\hat{\SIGMA}_n$, even though the latter estimator uses the full, undistorted samples. \ae{Moreover, the estimator $\SIGMADITH_n$ performs \rev{similarly} but slightly worse. The poorer performance can partially be explained by the fact it does not exploit that $\SIGMA$ has an all-ones diagonal. To illustrate this, we also plot the performance of $\hat{P}_{\text{PSD}}(\SIGMADITH_n)$, where $\hat{P}_{\text{PSD}}$ is the projection onto the set of positive semidefinite matrices with unit diagonal. Although this projection slightly improves the performance, it does not close the gap.}\par 
Overall, it seems that coarse quantization hardly causes a loss in estimation quality. Let us emphasize that a proper choice of $\lambda$ is essential for the \sjoerdgreen{performance of} $\SIGMADITH_n$, see Figure \ref{fig:Lambda}. Figure \ref{fig:LambdaComparison} furthermore suggests that the optimal choice of $\lambda$ is influenced by the number of samples $n$, supporting the dependence of $\lambda$ on $\log(n)$ in Theorem \ref{thm:OperatorDitheredMask}.

\begin{figure}[!htb]
	\centering
	\begin{subfigure}[c]{0.45\textwidth}
		\includegraphics[width=\textwidth]{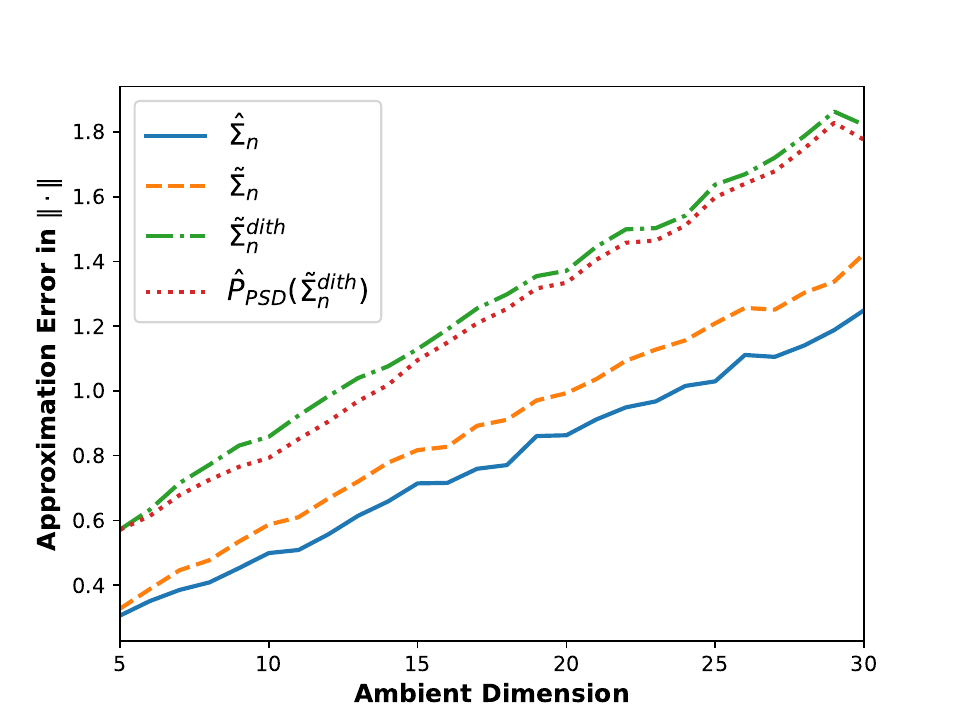}
		\subcaption{Varying $p$.}
		\label{fig:CompareAll}
	\end{subfigure} \quad
	\begin{subfigure}[c]{0.47\textwidth}
		\includegraphics[width=\textwidth]{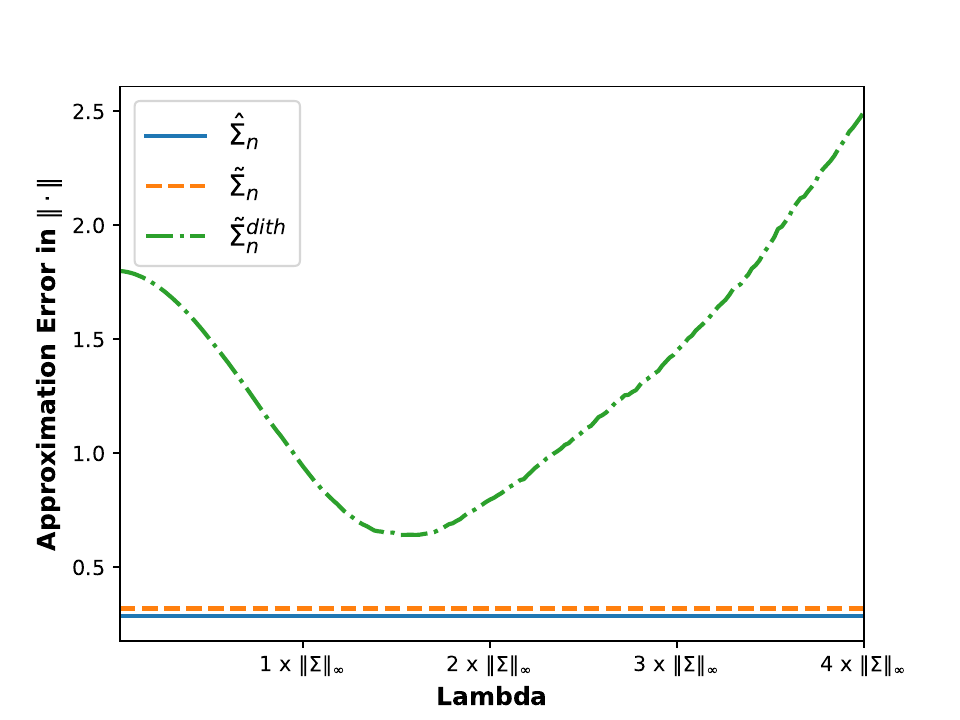}
		\subcaption{Varying $\lambda$.}
		\label{fig:Lambda}
	\end{subfigure}
	\caption{The left plot depicts \sjoerd{the} average estimation error\sjoerd{s} in operator norm, for $n = 200$ and $p$ varying from $5$ to $30$. The dithered estimator uses \sjoerd{the} $\lambda \in (0,4 \pnorm{\SIGMA}{\infty})$ \sjoerd{that is} optimized via grid-search.\\
		The right plot depicts \sjoerd{the} average estimation error in operator norm, for $n = 200$, $p=5$, and $\lambda$ varying from $0$ to $4 \pnorm{\SIGMA}{\infty}$. \sjoerd{Although they are} not affected by changes in $\lambda$, \sjoerd{the sample covariance matrix and $\tilde{\SIGMA}_n$} are \sjoerd{depicted} for reference.}
\end{figure}

\begin{figure}
	\centering
	\includegraphics[width=0.5\textwidth]{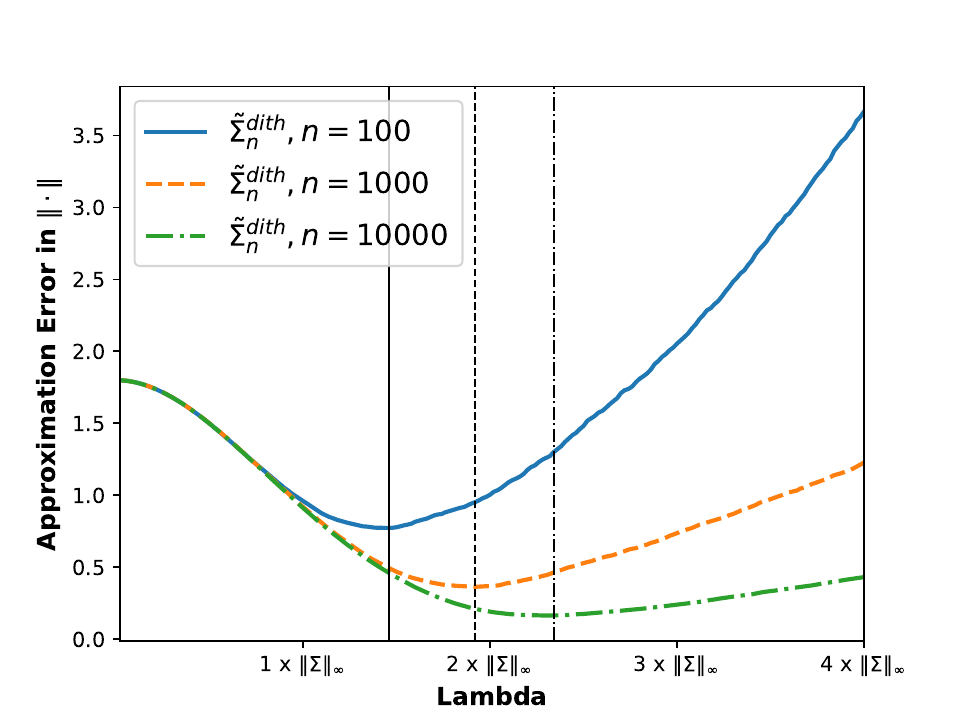}
	\caption{Optimal choices of $\lambda$ for various numbers of samples. Here $p=5$ is fixed.}
	\label{fig:LambdaComparison}
\end{figure}

\subsection{One-\sjoerd{b}it \sjoerd{e}stimator -- \sjoerd{i}nfluence of \sjoerd{c}orrelation}

In our second experiment we compare \sjoerd{the} sample covariance matrix $\hat{\SIGMA}_n$ \sjoerd{to the} estimator $\tilde{\SIGMA}_n$. \sjoerd{Theorem~\ref{thm:Operator}} suggests that the performance of $\tilde{\SIGMA}_n$ heavily depends on the correlations between the different entries of $\X$ (via the off-diagonal entries of $\sjoerdgreen{\A}$, \sjoerd{see the discussion following Theorem~\ref{thm:Operator}}). To \sjoerd{illustrate this numerically}, we choose three different ground truths $\SIGMA \in \R^{p\times p}$ having ones on the diagonal and being constant $c$ on all remaining entries: one with \ae{low correlation ($c = 0.1$)}, one with high correlation ($c = 0.9$), and one with very high correlation ($c = 0.99$). We fix the ambient dimension as $p = 20$ and vary the number of samples $\X_1,...,\X_n$ from $n = 10$ to $n=300$. As Figure \ref{fig:Correlation} illustrates, the estimator $\tilde{\SIGMA}_n$ indeed outperforms the sample covariance matrix if the correlation is high. \sjoerd{This is surprising, as the sample covariance matrix uses the full, rather than the quantized samples. Let us emphasize, though, that the sample covariance matrix will outperform $\tilde{\SIGMA}_n$ in many other \ae{scenarios. The low correlation setting $c=0.1$ is a concrete example}. In particular, $\tilde{\SIGMA}_n$ is only a suitable estimator if \ae{$\operatorname{diag}(\SIGMA)$ consists only of ones.}}

\begin{figure}[!htb]
	\centering
	\begin{subfigure}[c]{0.45\textwidth}
		\includegraphics[width=\textwidth]{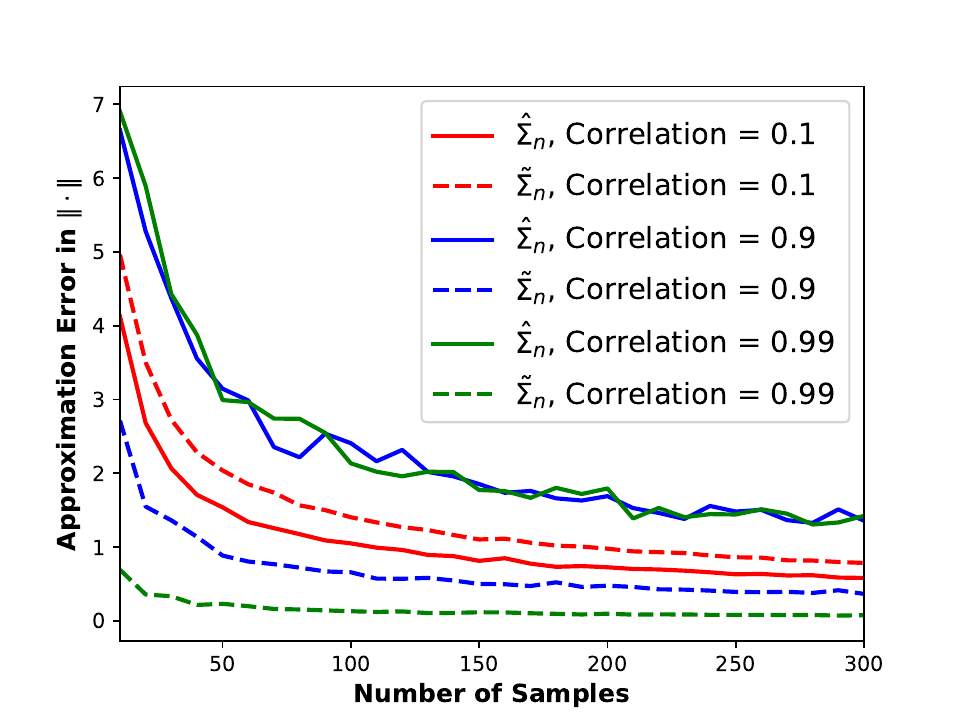}
		\subcaption{$\hat{\SIGMA}_n$ vs $\tilde{\SIGMA}_n$.}
		\label{fig:Correlation}
	\end{subfigure} \\
	\begin{subfigure}[c]{0.45\textwidth}
		\includegraphics[width=\textwidth]{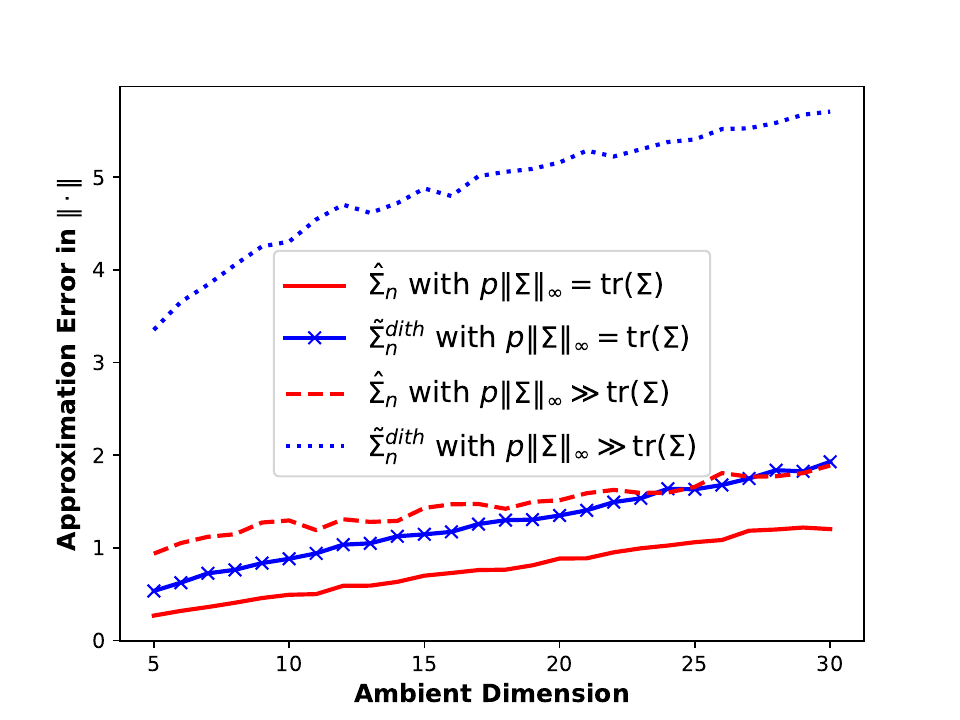}
		\subcaption{$\hat{\SIGMA}_n$ vs $\tilde{\SIGMA}_n^{\text{dith}}$.}
		\label{fig:Diagonal}
	\end{subfigure}
	\quad
	\begin{subfigure}[c]{0.45\textwidth}
		\includegraphics[width=\textwidth]{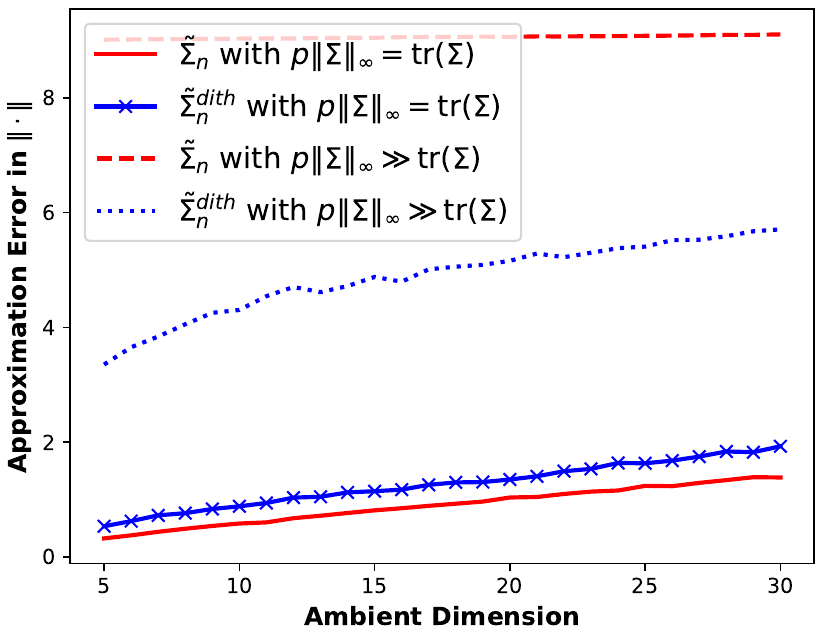}
		\subcaption{$\tilde{\SIGMA}_n$ vs $\tilde{\SIGMA}_n^{\text{dith}}$.}
		\label{fig:BiasedVSDithered}
	\end{subfigure}
	\caption{ The \ae{upper} plot depicts the average estimation errors of $\hat{\SIGMA}_n$ and $\tilde{\SIGMA}_n$ in operator norm, for $p = 20$, $n$ varying from $10$ to $300$ and three different choices of the ground truth $\SIGMA$ with ones on the diagonal and off-diagonal entries equal to \ae{$c = 0.1$}, $c = 0.9$, and $c = 0.99$.\\
		The \ae{lower left} plot depicts the average estimation errors of $\hat{\SIGMA}_n$ and $\tilde{\SIGMA}_n^{\text{dith}}$ in operator norm, for $n = 200$, $p$~varying from $5$ to $30$ and two different choices of the ground truth $\SIGMA$. \\
		The \ae{lower right} plot depicts the average estimation errors of $\tilde{\SIGMA}_n$ and $\tilde{\SIGMA}_n^{\text{dith}}$ in operator norm, for $n = 200$, $p$~varying from $5$ to $30$ and two different choices of the ground truth $\SIGMA$.}
\end{figure}

\subsection{Two-\sjoerd{b}it \sjoerd{e}stimator -- \sjoerd{i}nfluence of \sjoerd{d}iagonal}

In our third experiment we verify \sjoerd{that there is a performance gap between the} sample covariance matrix $\hat{\SIGMA}_n$ and the dithered two-bit estimator $\SIGMADITH_n$ for covariance matrices $\SIGMA$ with $\tr(\SIGMA) \ll p \pnorm{\SIGMA}{\infty}$, \sjoerd{see} Theorem \ref{thm:OperatorDitheredMask} and the subsequent discussion. To this end, we compare the reconstruction of $\SIGMA$ having ones on the diagonal and entries $0.2$ everywhere else ($\tr(\SIGMA) = p \pnorm{\SIGMA}{\infty}$) with the reconstruction of $\SIGMA'$ being the same as $\SIGMA$ apart from $\Sigma'_{1,1} = 10$ ($\tr(\SIGMA) \ll p \pnorm{\SIGMA}{\infty}$). We fix the number of samples as $n = 200$ and vary $p \in [5,30]$. Figure \ref{fig:Diagonal} shows a considerable increase of the gap in reconstruction accuracy in the second \sjoerd{case.} 

\ae{
\subsection{One-bit vs two-bit estimator}
Let us conclude the numerical section by a brief comparison of $\tilde\SIGMA_n$ and $\SIGMADITH_n$. Figure \ref{fig:CompareAll} shows that $\tilde\SIGMA_n$ clearly outperforms $\SIGMADITH_n$ if $\SIGMA$ is constant on its diagonal. At first sight, this seems to be counter-intuitive since $\SIGMADITH_n$ uses twice as many samples. Due to its specific shape, the estimator $\tilde\SIGMA_n$ is however perfectly tuned to a constant-diagonal setting. Figure \ref{fig:BiasedVSDithered} shows that $\SIGMADITH_n$ outperforms $\tilde\SIGMA_n$ in the problem of estimating a general $\SIGMA$. Evidently, it is impossible to compare $\tilde\SIGMA_n$ and $\SIGMADITH_n$ on even terms. Either $\tilde\SIGMA_n$ is favoured by its intrinsic constant-diagonal regularization or $\SIGMADITH_n$ is favoured by its capability of estimating non-constant diagonals. Compared to these factors, the difference in the number of samples used to compute these estimators plays a minor role.
}


\section{Conclusion}
\label{sec:Discussion}

\ae{In the present paper, we examined the impact of coarse sample quantization on the estimation of covariance and correlation matrices. We analyzed a one-bit estimator that is limited to the estimation of correlation matrices of centered Gaussian distributions by design and proposed a novel two-bit estimator that allows to estimate full covariance matrices of subgaussian distributions. We provided near-optimal guarantees for both estimators that match corresponding minimax lower bounds. We furthermore identified special settings in which our quantized estimators are outperformed by or outperform the classical sample covariance matrix that uses full knowledge of the samples. As our exposition shows, it is hard to compare the one-bit and two-bit estimators directly; each has its particular strengths and weaknesses and there is no common ground for fair comparison.} 


\appendix

\section{\sjoerd{M}inimax \sjoerd{l}ower \sjoerd{b}ounds}

\ae{In this appendix, we derive the previously mentioned minimax} \sjoerd{lower} bounds for the quantized covariance estimation setting \ae{and compare them to the upper bounds in our main results. We restrict ourselves to the unmasked case, i.e., $\M = \mathbf 1$. Consequently, we are only interested in estimation rates for $n > p$.} In the case of quantization without dithering, a minimax lower bound is directly implied by the results in \cite{cai2010optimal}. \ae{This can be seen as follows. Let us assume for simplicity \sjoerd{that} $p$ is even and \sjoerd{consider} the set of covariance matrices
\begin{align*}
    \mathcal{F}^* & = \curly{ \SIGMA_{\ttheta} = \id + \frac{1}{2\sqrt{np}} \sum_{k = 1}^{\frac{p}{2}} \theta_k \M_k  \colon \ttheta \in \curly{0,1}^{\frac{p}{2}}  } \\
    & \subset \curly{ \SIGMA \in \R^{p\times p} \colon \Sigma_{i,i} = 1, \text{ for all } i \in [p]},
\end{align*}
where the matrices $\M_k \in \R^{p\times p}$ have entries
\begin{align*}
    (\M_k)_{i,j} = 1_{\{ i = k \text{ and } k+1 \le j \le p, \text{ or } j = k \text{ and } k+1 \le i \le p \}}.
\end{align*}
The partial result \cite[Equation (33)]{cai2010optimal} implies that
\begin{align} \label{eq:MinimaxBound}
	\inf_{\hat{\SIGMA}} \sup_{\SIGMA \in \mathcal{F}^*} \johannes{\mathbb{E} \pnorm{ \hat{\SIGMA} - \SIGMA }{}^2}{}  \ge c \frac{p}{n},
\end{align}
where the infimum is taken over all estimators that are based on i.i.d.\ observations from an underlying multivariate normal distribution. Our one-bit estimator $\tilde{\SIGMA}_n$ clearly falls into this category. 

Let us now consider quantization with dithering. In this case the estimator is not solely based on i.i.d.\ observations of a normal distribution but also depends on the dithering vectors. Nevertheless, we can reduce the dithered setting to the one in \cite{cai2010optimal} by using the following lemma. We require some additional notation before stating the result.} For a Gaussian distribution with zero mean and covariance matrix $\SIGMA \in \R^{p\times p}$, we denote the corresponding probability \sjoerd{distribution} by $\PP_{\SIGMA}$ and its density by $\phi_{\SIGMA}$, i.e.,
\begin{align*}
\phi_{\SIGMA}(\z) = \frac{1}{\sqrt{(2\pi)^p \det(\SIGMA)}} e^{-\frac{1}{2}  \z^T \SIGMA^{-1}\z}.
\end{align*} 
\sjoerd{If $\X \sim \Nc(\0,\SIGMA)$, $\lambda>0$, $\Tau \sim \text{Unif}([-\lambda,\lambda]^p)$, and $\Y = \sign(\X + \Tau)$, then we denote the probability distribution of $\Y$ by $\QQ_{\SIGMA}$ and use $q_{\SIGMA}$ to denote the associated probability mass function.} 
\begin{lemma} \label{lem:TVrelation}
	For any two covariance matrices $\SIGMA, \SIGMA' \in \R^{p\times p}$,
	\begin{align*}
	\sum_{\y \in \{-1,1\}^p} \min \{ q_{\SIGMA} (\y) , q_{\SIGMA'} (\y) \} \ge \int_{\R^p} \min \{ \phi_{\SIGMA}(\x) , \phi_{\SIGMA'}(\x) \} \dx{\x}.
	\end{align*}
\end{lemma}
\begin{proof}
	\sjoerd{Since $\min\{a,b\}=\frac{1}{2}(a+b)-\frac{1}{2}|a-b|$ for all $a,b\in\R$, it suffices to show that 
		$$\sum_{\y \in \{-1,1\}^p} \abs{q_{\SIGMA} (\y) - q_{\SIGMA'} (\y) }{} \leq \int_{\R^p} \abs{\phi_{\SIGMA} (\z) - \phi_{\SIGMA'} (\z)}{} \dx{\z}.$$
		Observe that
		\begin{align*}
		q_{\SIGMA} (\y) = \int_{[-\lambda,\lambda]^p} \round{\frac{1}{2\lambda}}^p \int_{O_{\y}^{\Tau}} \phi_{\SIGMA} (\z) \dx{\sjoerdgreen{\z}} \dx{\Tau},
		\end{align*}
		where $O_{\y}^{\Tau} = \{ \z \in \R^p \colon \sign(\z +\Tau) = \y \}$.} Consequently,
	\begin{align*}
	&\sum_{\y \in \{-1,1\}^p} \abs{q_{\SIGMA} (\y) - q_{\SIGMA'} (\y) }{} \\
	&\quad= \sum_{\y \in \{-1,1\}^p} \abs{\int_{[-\lambda,\lambda]^p} \round{\frac{1}{2\lambda}}^p \int_{O_{\y}^{\Tau}} \phi_{\SIGMA} (\z) \dx{\z} \dx{\Tau} - \int_{[-\lambda,\lambda]^p} \round{\frac{1}{2\lambda}}^p \int_{O_{\y}^{\Tau}} \phi_{\SIGMA'} (\z) \dx{\z} \dx{\Tau} }{} \\
	&\quad\le \sum_{\y \in \{-1,1\}^p} \int_{[-\lambda,\lambda]^p} \round{\frac{1}{2\lambda}}^p \int_{O_{\y}^{\Tau}} \abs{\phi_{\SIGMA} (\z) - \phi_{\SIGMA'} (\z)}{} \dx{\z} \dx{\Tau} \\
	&\quad= \int_{[-\lambda,\lambda]^p} \round{\frac{1}{2\lambda}}^p \int_{\R^p} \abs{\phi_{\SIGMA} (\z) - \phi_{\SIGMA'} (\z)}{} \dx{\z} \dx{\Tau} 
	= \int_{\R^p} \abs{\phi_{\SIGMA} (\z) - \phi_{\SIGMA'} (\z)}{} \dx{\z}.
	\end{align*}
\end{proof}
\ae{Using Lemma~\ref{lem:TVrelation}, we can easily extend the minimax lower bound in \eqref{eq:MinimaxBound} to our dithered observations. To see this, one follows the steps in \cite[Section 3.3.1]{cai2010optimal} but uses Lemma~\ref{lem:TVrelation} to generalize \cite[Lemma 6]{cai2010optimal} from $\X \sim \Nc(\0,\SIGMA)$ to $\Y = \sign(\X+\Tau)$. \sjoerd{Note} that \cite[Lemma 5]{cai2010optimal} applies to $\mathcal{F}^*$, for any $n > p^{1+2\alpha}$ where $\alpha > 0$ may be chosen arbitrarily small.}

\ae{
\subsection{Comparing the bounds}

Observe that any $\SIGMA\in \mathcal F^*$ has unit diagonal and satisfies $\|\SIGMA\| \le 2$. If $n\geq p\log(p)$, then for any $\SIGMA$ with $\|\SIGMA\|\leq 2$, with high probability
\begin{align*}
	\| \SIGMADITH_n - \SIGMA \|^2 \lesssim \log(n)^2 \frac{p \log(p)}{n}.
\end{align*}
Thus, up to log-factors the estimator $\SIGMADITH_n$ achieves the minimax optimal rate on the set of covariance matrices with operator norm at most $2$.\par
Let us now show that $\tilde{\SIGMA}_n$ achieves the minimax optimal rate up to log-factors on the set of covariance matrices with unit diagonal and operator norm at most $2$. Consider any $\SIGMA$ with these properties. We start by estimating the second term on the right hand side of the estimate \eqref{eqn:OperatorEst} in Theorem~\ref{thm:Operator}. Recall that $\A$ has entries
$$A_{ij}=\sqrt{1-\Sigma_{i,j}^2}.$$
Using the series expansion
$$\sqrt{1+z}=\sum_{k=0}^{\infty} {1/2 \choose k} z^k, \qquad |z|\leq 1,$$
we can write 
\begin{align*}
\|\A\| & = \left\|\sum_{k=0}^{\infty} {1/2 \choose k} (-1)^k \SIGMA^{\odot 2k}\right\| \\
& \leq \|\boldsymbol{1}\| + \sum_{k=1}^{\infty} \left|{1/2 \choose k}\right| \|\SIGMA^{\odot 2k}\| \leq p + \sum_{k=1}^{\infty} \left|{1/2 \choose k}\right|,
\end{align*}
where we used that $\|\SIGMA^{\odot 2k}\|\leq \|\SIGMA\|_{\infty}^{2k-1}\|\SIGMA\|\lesssim 1$ by Lemma~\ref{lem:HadamardOperatorNorm}. By Stirling's approximation, there is a constant $C>0$ and $K\in \mathbb{N}$ such that
$$\left|{1/2 \choose k}\right| = \frac{(2k)!}{(2^k k!)^2(2k-1)}\leq \frac{C}{k^{3/2}}$$
for $k\geq K$. Hence, it follows that 
$$\max\{\|\A\|,\|\SIGMA\|\}\lesssim p.$$
To estimate the first term on the right hand side of \eqref{eqn:OperatorEst}, recall that 
\begin{align*}
    \| \sigma(\A) \|^2 
    = \Big\| \A^2 \odot \frac{2}{\pi} \arcsin(\SIGMA) - \Big(\A \odot \frac{2}{\pi} \arcsin(\SIGMA) \Big)^2 \Big\|,
\end{align*}
which, by \eqref{eq:Kadison} and Lemma \ref{lem:HadamardOperatorNorm}, satisfies
\begin{equation} \label{eq:SigmaBound}
    \| \sigma(\A) \|^2 \lesssim \Big\| \A^2 \odot \frac{2}{\pi} \arcsin(\SIGMA)\Big\| \lesssim \|\A^2\|_{\infty}\|\arcsin(\SIGMA)\|.
\end{equation}
Using the series expansion
$$\arcsin(z) = \sum_{k=0}^{\infty}\frac{(2k)!}{(2^k k!)^2(2k+1)}z^{2k+1}, \qquad |z|\leq 1,$$
we find using $\|\SIGMA\|_{\infty}\leq 1$ that
\begin{align} \label{eq:MinimaxPartI}
\begin{split}
    \| \arcsin(\SIGMA) \|
    &\le \sum_{k=0}^{\infty} \frac{(2k)!}{(2^k k!)^2(2k+1)} \| \SIGMA^{\odot (2k+1)} \| \\
    &\le \| \SIGMA \| \sum_{k=0}^{\infty} \frac{(2k)!}{(2^k k!)^2(2k+1)} \| \SIGMA \|_\infty^{2k} \lesssim \arcsin(1) = \frac{\pi}{2},
\end{split}
\end{align}
Moreover,
\begin{align}
\begin{split} \label{eq:MinimaxPartII}
    \| \A^2 \|_\infty 
    = \max_{1 \le i \le p} \| \A\e_i \|_2^2
    \le p \| \A \|_\infty^2 \le p,
\end{split}
\end{align}
where $\e_i$ denotes the $i$-th canonical basis vector. Combining these estimates, we see that $\|\sigma(\A)\|\lesssim\sqrt{p}$. Hence, \eqref{eqn:OperatorEst} yields that with high probability
	\begin{equation*}
	    \pnorm{\tilde{\SIGMA}_n - \SIGMA}{} 
	\lesssim  \sjoerd{\|\sigma(\A)\|}   \sqrt{\frac{\sjoerd{\log(p)}}{n}} + \max\curly{ \pnorm{\A}{}, \pnorm{\SIGMA}{} }  \frac{\sjoerd{\log(p)} }{n} 
	 \lesssim \sqrt{\frac{\sjoerd{p\log(p)}}{n}}
	\end{equation*}
if $n\geq p\log(p)$. 


}


\section*{Acknowledgements}
 
The authors were supported by the Deutsche Forschungsgemeinschaft (DFG, German Research Foundation) through the project CoCoMIMO funded within the priority program SPP 1798 Compressed Sensing in Information Processing (COSIP).

\bibliography{mybib}{}
\bibliographystyle{plain}

\end{document}